\documentclass[11.5pt,reqno]{amsart}
\usepackage{graphicx,amssymb,amsmath,amsthm}
\usepackage{comment}
\usepackage{subcaption}
\usepackage{cases}
\usepackage{enumitem}
\usepackage[colorlinks,
linkcolor=red,
anchorcolor=blue,
citecolor=green,
urlcolor=blue,
breaklinks=true
]{hyperref}

\allowdisplaybreaks[4]

\newcommand{\ii}{{\rm i}}
\newcommand{\nv}{{\boldsymbol{\nu}}}
\allowdisplaybreaks[4]

\makeatletter
\DeclareFontFamily{U}{tipa}{}
\DeclareFontShape{U}{tipa}{m}{n}{<->tipa10}{}
\newcommand{\arc@char}{{\usefont{U}{tipa}{m}{n}\symbol{62}}}%
\newcommand{\arc}[1]{\mathpalette\arc@arc{#1}}
\newcommand{\arc@arc}[2]{%
\sbox0{$\m@th#1#2$}%
\vbox{
\hbox{\resizebox{\wd0}{\height}{\arc@char}}
\nointerlineskip
\box0
}%
}
\makeatother

\newtheorem{theorem}{Theorem}[section]
\newtheorem{lemma}{Lemma}[section]
\newtheorem{corollary}{Corollary}[section]

\newtheorem{remark}{Remark}[section]
\newtheorem{definition}{Definition}[section]
\newtheorem{question}{Question}[section]
\newtheorem{problem}{Problem}[section]
\newtheorem{conjecture}{Conjecture}[section]

\numberwithin{equation}{section}

\begin{document}

\title[Optimal Frames For Phase Retrieval via Optimal Polygons]{Optimal Frames for Phase Retrieval from Edge Vectors of Optimal Polygons}

\author{Zhiqiang Xu}
\address{State Key Laboratory of Mathematical Sciences, Academy of Mathematics and Systems Science, Chinese Academy of Sciences, Beijing, 100190, China; School of Mathematical Sciences, University of Chinese Academy of Sciences, Beijing, 100091, China}
\email{xuzq@lsec.cc.ac.cn}
\thanks{The first author was supported by the National Science Fund for Distinguished Young Scholars (12025108) and NSFC grant (12471361, 12021001, 12288201).}

\author{Zili Xu}
\address{School of Mathematical Sciences, Key Laboratory of MEA(Ministry of Education) \& Shanghai Key Laboratory of PMMP, East China Normal University, Shanghai, 200241, China}
\email{zlxu@math.ecnu.edu.cn}
\thanks{The second author was supported by NSFC grant (12501121, 12571105) and the Natural Science Foundation of Shanghai grant (25ZR1402131).}

\author{Xinyue Zhang}
\address{State Key Laboratory of Mathematical Sciences, Academy of Mathematics and Systems Science, Chinese Academy of Sciences, Beijing, 100190, China; School of Mathematical Sciences, University of Chinese Academy of Sciences, Beijing, 100091, China}
\email{zhangxinyue@amss.ac.cn}
\thanks{The third author is the corresponding author.}





\begin{abstract}

This paper aims to characterize the optimal frame for phase retrieval, defined as the frame whose condition number for phase retrieval attains its minimal value. In the context of the two-dimensional real case, we reveal the connection between optimal frames for phase retrieval and the perimeter-maximizing isodiametric problem, originally proposed by Reinhardt in 1922. Our work establishes that every optimal solution to the perimeter-maximizing isodiametric problem inherently leads to an optimal frame in ${\mathbb R}^2$. 
By recasting the optimal polygons problem as one concerning the discrepancy of roots of unity, we characterize all optimal polygons. Building upon this connection, we then characterize all optimal frames with $m$ vectors in ${\mathbb R}^2$ for phase retrieval when $m \geq 3$ has an odd factor.
As a key corollary, we show that the harmonic frame $E_m \subset {\mathbb R}^2$ is {\em not} optimal for any even integer $m \geq 4$. This finding disproves a conjecture proposed by Xia, Xu, and Xu [{\em Math. Comp.}, 94 (2025), pp.~2931--2960]. Previous work has established that $E_m$ is indeed optimal when $m$ is an odd integer. 

\end{abstract}

\maketitle

\section{Introduction}
\setcounter{equation}{0}

\subsection{Problem setup}
Let $\boldsymbol{A} = [\boldsymbol{a}_1, \ldots, \boldsymbol{a}_m]^T\in \mathbb{H}^{m \times d}$ be a given measurement matrix, where each $\boldsymbol{a}_j \in \mathbb{H}^d$ and $\mathbb{H} \in \{\mathbb{R}, \mathbb{C}\}$. 
The phase retrieval problem seeks to recover an unknown signal $\boldsymbol{x} \in \mathbb{H}^d$ from its phaseless measurements $\left|\left\langle \boldsymbol{a}_j, \boldsymbol{x} \right\rangle\right|$, $j = 1, \ldots, m$. 
Define the nonlinear measurement map $\Phi_{\boldsymbol{A}}: \mathbb{H}^d \rightarrow \mathbb{R}_{+}^m$ by
\begin{equation}\label{PhiA}
	\Phi_{\boldsymbol{A}}(\boldsymbol{x})=|\boldsymbol{A} \boldsymbol{x}|:=\left(\left|\left\langle\boldsymbol{a}_1, \boldsymbol{x}\right\rangle\right|,\left|\left\langle\boldsymbol{a}_2, \boldsymbol{x}\right\rangle\right|, \ldots,\left|\left\langle\boldsymbol{a}_m, \boldsymbol{x}\right\rangle\right|\right)^T 
	\in \mathbb{R}_{+}^m.
\end{equation}
A matrix $\boldsymbol{A}$ is said to have the phase retrieval property if $|\boldsymbol{A} \boldsymbol{x}|=|\boldsymbol{A} \boldsymbol{y}|$ implies $\boldsymbol{x}=c \cdot \boldsymbol{y}$ for some $c \in \mathbb{H}$ with $|c|=1$. 
It is known that $m \geq 2 d-1$  (when $\mathbb{H} = \mathbb{R}$) or $m \geq 4 d-4$ (when $\mathbb{H} = \mathbb{C}$) generic measurements suffice for phase retrieval property \cite{balan2006signal, bandeira2014saving, conca2015algebraic, wang2019generalized}.

A key metric for evaluating the robustness of the phase retrieval is the condition number. 
For any $\boldsymbol{x},\boldsymbol{y}\in \mathbb{H}^d$, we define the distance between $\boldsymbol{x}$ and $\boldsymbol{y}$ as  
\begin{equation*}
	\operatorname{dist}_{\mathbb{H}}(\boldsymbol{x}, \boldsymbol{y}):=\min \left\{\|\boldsymbol{x}-c \cdot \boldsymbol{y}\|_2: c \in \mathbb{H},|c|=1\right\},
\end{equation*}
where $\|\cdot \|_2$ represents the Euclidean norm.
The condition number $\beta_{\boldsymbol{A}}^{\mathbb{H}}$ of a given measurement matrix ${\boldsymbol{A}}\in \mathbb{H}^{m\times d}$
is defined as
\begin{equation*}
	\beta_{\boldsymbol{A}}^{\mathbb{H}}:=\frac{U_{\boldsymbol{A}}^{\mathbb{H}}}{L_{\boldsymbol{A}}^{\mathbb{H}}},
\end{equation*}
where 
$L_{\boldsymbol{A}}^{\mathbb{H}}$ and $U_{\boldsymbol{A}}^{\mathbb{H}}$ represent the optimal lower and upper Lipschitz constants of the map $\Phi_{\boldsymbol{A}}$, i.e.,
\begin{equation*}
	L_{\boldsymbol{A}}^{\mathbb{H}}
	:=\inf _{\substack{\boldsymbol{x}, \boldsymbol{y} \in \mathbb{H}^d \\ \operatorname{dist}_{\mathbb{H}}(\boldsymbol{x}, \boldsymbol{y}) \neq 0}} \frac{\||\boldsymbol{A} \boldsymbol{x}|-|\boldsymbol{A} \boldsymbol{y}|\|_2}{\operatorname{dist}_{\mathbb{H}}(\boldsymbol{x}, \boldsymbol{y})} 
	\quad \text { and } \quad 
	U_{\boldsymbol{A}}^{\mathbb{H}}
	:=\sup _{\substack{\boldsymbol{x}, \boldsymbol{y} \in \mathbb{H}^d \\ \operatorname{dist}_{\mathbb{H}}(\boldsymbol{x}, \boldsymbol{y}) \neq 0}} \frac{\||\boldsymbol{A} \boldsymbol{x}|-|\boldsymbol{A} \boldsymbol{y}|\|_2}{\operatorname{dist}_{\mathbb{H}}(\boldsymbol{x}, \boldsymbol{y})}.
\end{equation*}
A smaller value of $\beta_{\boldsymbol{A}}^{\mathbb{H}}$ entails that $\Phi_{\boldsymbol{A}}$ behaves more like a near-isometry, indicating greater stability against measurement noise.
If $\boldsymbol{A}$ lacks the phase retrieval property, then $L_{\boldsymbol{A}}^{\mathbb{H}}=0$ and $\beta_{\boldsymbol{A}}^{\mathbb{H}}=+\infty$. 
We say that a matrix $\boldsymbol{A}\in\mathbb{H}^{m\times d}$ has {\em the minimal condition number} if $\beta_{\boldsymbol{A}}^{\mathbb{H}}=\min _{\boldsymbol{M} \in \mathbb{H}^{m \times d}} \beta_{\boldsymbol{M}}^{\mathbb{H}}$.
To simplify the notation, we often omit the superscript in $\beta_{\boldsymbol{A}}^{\mathbb{H}}$, determining whether $\beta_{\boldsymbol{A}}^{\mathbb{H}}$ is defined over the real or complex field based on whether the matrix $\boldsymbol{A}$ is real or complex. 
In the same way, we drop the superscripts in $L_{\boldsymbol{A}}^{\mathbb{H}}, U_{\boldsymbol{A}}^{\mathbb{H}}$, and we omit the subscript in $\operatorname{dist}_{\mathbb{H}}(\boldsymbol{x}, \boldsymbol{y})$ when the underlying field is clear from context.

Recently, the authors in \cite{xia2024stability} derive the first universal constant lower bound on the condition number $\beta_{\boldsymbol{A}}^{\mathbb{H}}$ for all $\boldsymbol{A} \in \mathbb{H}^{m \times d}$, namely,
\begin{equation}
	\label{beta0}
	\beta_{\boldsymbol{A}}^{\mathbb{H}} \geq \beta_0^{\mathbb{H}}:= \begin{cases}\sqrt{\frac{\pi}{\pi-2}} \approx 1.659 & \text { if } \mathbb{H}=\mathbb{R}, \\ \sqrt{\frac{4}{4-\pi}} \approx 2.159 & \text { if } \mathbb{H}=\mathbb{C}.\end{cases}
\end{equation}
The condition number of a standard Gaussian matrix in $\mathbb{H}^{m \times d}$ asymptotically matches the lower bound $\beta_0^{\mathbb{H}}$ as $m \rightarrow \infty$ \cite{xia2024stability}. 
Furthermore, in the real case $\mathbb{H}=\mathbb{R}$, they show that for all $\boldsymbol{A} \in \mathbb{R}^{m \times d}$, 
\begin{equation}
	\label{beta_HF}
	\beta_{\boldsymbol{A}} \geq \frac{1}{\sqrt{1-\frac{1}{m \cdot \sin \frac{\pi}{2 m}}}}.
\end{equation}
In particular, they provide that 
\begin{equation}\label{betaA-R2}
	\beta_{\boldsymbol{E}_m}=\begin{cases}
		\frac{1}{\sqrt{1-\frac{2}{m\cdot \sin\frac{\pi}{m}}}}  & \text{if $m\geq 3$ is even,}\\	
		\frac{1}{\sqrt{1-\frac{1}{m\cdot \sin\frac{\pi}{2m}}}} & \text{if $m\geq 3$ is odd,}
	\end{cases}	
\end{equation}
where the rows of the matrix $\boldsymbol{E}_m$ form the harmonic frame in ${\mathbb R}^2$, i.e.,  
\begin{equation}
	\label{HF}
	\boldsymbol{E}_m:=\left(\begin{array}{cccc}
		1 & \cos  \frac{1}{m} \pi  & \cdots & \cos  \frac{m-1}{m} \pi  \\
		0 & \sin  \frac{1}{m} \pi  & \cdots & \sin  \frac{m-1}{m} \pi 
	\end{array}\right)^{T} \in \mathbb{R}^{m \times 2}.
\end{equation}
Combining \eqref{beta_HF} and \eqref{betaA-R2}, they further obtain that $\beta_{\boldsymbol{E}_m}=\min _{\boldsymbol{A} \in \mathbb{R}^{m \times 2}} \beta_{\boldsymbol{A}}$ for each {\em odd} integer $m \geq 3$.
Consequently, the following conjecture is naturally proposed in \cite{xia2024stability} for each even integer $m\geq 4$.
\begin{conjecture}{\rm \cite{xia2024stability}}
	\label{meven}
	If $m \geq 4$ is even, then 
	\[
	\beta_{\boldsymbol{E}_m}=\min _{\boldsymbol{A} \in \mathbb{R}^{m \times 2}} \beta_{\boldsymbol{A}}.
	\]	
\end{conjecture}


\subsection{Our contribution}


In this paper, we will disprove Conjecture \ref{meven} by deriving a sharper estimate for the minimal condition number of matrices $\boldsymbol{A}\in \mathbb{R}^{m \times 2}$.
Our key observation is that minimizing the condition number of $\boldsymbol{A}\in \mathbb{R}^{m \times 2}$ is equivalent to solving a classical problem in discrete geometry called the perimeter-maximizing isodiametric problem, which was introduced by Reinhardt in 1922 \cite{reinhardt1922extremale}.
This equivalence leads directly to an improved bound on the minimal condition number, which we state in the following theorem.
\begin{theorem}\label{mincond}
	\begin{enumerate}
		\item[\rm (i)] 
		For any integer $m\geq 3$ that has an odd factor, we have 
		\begin{equation*}
			\min _{\boldsymbol{A} \in \mathbb{R}^{m \times 2}} \beta_{\boldsymbol{A}} = \frac{1}{\sqrt{1-\frac{1}{m \cdot \sin \frac{\pi}{2 m}}}}.
		\end{equation*} 
		
		\item[\rm (ii)]
		For $m=2^s$, where $s$ is an integer with $s\geq 2$, we have 
		\begin{equation*}
			\min _{\boldsymbol{A} \in \mathbb{R}^{m \times 2}} \beta_{\boldsymbol{A}} > \frac{1}{\sqrt{1-\frac{1}{m \cdot \sin \frac{\pi}{2 m}}}}.
		\end{equation*} 
		
		\item[\rm (iii)] The harmonic frame $\boldsymbol{E}_m$  in \eqref{HF} has the minimal condition number if and only if $m$ is odd. 
		Moreover,  $\boldsymbol{E}_m$ is the unique matrix with minimal condition number if and only if $m$ is a odd prime. 
	\end{enumerate}
\end{theorem}

\begin{remark}
	The authors in \cite{xia2024stability} proved that the lower bound  \eqref{beta_HF} is tight if $m\geq 3$ is an odd integer.
	Theorem \ref{mincond} {\rm (i)} generalizes this result by showing that the lower bound \eqref{beta_HF} is tight as long as $m\geq 3$ has an odd factor.
	Moreover, Theorem \ref{mincond} {\rm (iii)} shows that Conjecture \ref{meven} is false. 
		For instance, when $m=4$, we will demonstrate in Remark \ref{example-m4} that
		\begin{equation*}
			\beta_{\boldsymbol{A}}
			=\min_{\boldsymbol{M}\in\mathbb{R}^{4\times 2}}\beta_{\boldsymbol{M}}
			=\sqrt{1+\frac{\sqrt{2}}{2}+\frac{\sqrt{6}}{2}}
			\approx 1.71,
		\end{equation*}
		where
		\begin{equation}\label{example-m=4}
			\boldsymbol{A}=
			\left(
			\begin{array}{cccc}
				-1 & \cos \frac{2\pi}{3}  & \sqrt{2\sin \frac{\pi}{12} }\cos \frac{3\pi}{8}  & \sqrt{2\sin \frac{\pi}{12} }\cos \frac{7\pi}{24}  \\
				0 & \sin \frac{2\pi}{3}  & \sqrt{2\sin \frac{\pi}{12} }\sin \frac{3\pi}{8}  & \sqrt{2\sin \frac{\pi}{12} }\sin \frac{7\pi}{24} 
			\end{array}
			\right)^{T}.
		\end{equation}
		In contrast, by \eqref{betaA-R2} we have $\beta_{\boldsymbol{E}_4}
		={\sqrt{2+\sqrt{2}}}
		\approx 1.84$, showing that $\beta_{\boldsymbol{E}_4}$ is not minimal.
\end{remark}

We next characterize all matrices in $\mathbb{R}^{m\times 2}$ with the minimal condition number, where $m \geq 3$ has an odd factor. For convenience, we introduce the following definitions.
\begin{definition}
	A finite set in $\mathbb{R}^d$ is called a \emph{frame} if it spans $\mathbb{R}^d$.
	We say that a frame $\{\boldsymbol{a}_1,\dotsc,\boldsymbol{a}_{m}\} \subset \mathbb{R}^d$ is \emph{optimal} if its associated matrix $\boldsymbol{A}=[\boldsymbol{a}_1,\ldots, \boldsymbol{a}_m]^T \in \mathbb{R}^{m \times d}$ has the minimal condition number, i.e., $\beta_{\boldsymbol{A}}=\min _{\boldsymbol{M} \in \mathbb{R}^{m \times d}} \beta_{\boldsymbol{M}}$. 
	A frame $\{\boldsymbol{a}_1,\dotsc,\boldsymbol{a}_{m}\}\subset\mathbb{R}^d$ is called a tight frame if $\sum_{i=1}^{m}\boldsymbol{a}_i\boldsymbol{a}_i^T=C\cdot \boldsymbol{I}_d$ for some constant $C>0$.
\end{definition}

\begin{definition}
	Denote by $\mathcal{T}$ the upper half-plane in $\mathbb{R}^2$ excluding the positive $x$-axis, i.e.,
	\begin{equation}\label{def-M-set}
		\mathcal{T}:=\{(x,y)^T\in\mathbb{R}^2: x=t\cdot \cos{\phi},y=t\cdot \sin{\phi},\,\, t\geq 0, \phi\in(0,\pi] \}.	
	\end{equation}
	For any vector $\nv=t\cdot (\cos\phi,\sin\phi)^T\in\mathcal{T}$ with $t\geq 0$ and $\phi\in (0,\frac{\pi}{2}]$, we define $\nv^{\perp}:=t\cdot (\cos(\phi+\frac{\pi}{2}),\sin(\phi+\frac{\pi}{2}))^T \in \mathcal{T}$.
\end{definition}

Note that the condition number of a matrix $\boldsymbol{A}=[\boldsymbol{a}_1,\ldots, \boldsymbol{a}_m]^T \in \mathbb{R}^{m \times 2}$ does not change if the row vector $\boldsymbol{a}_i$ is replaced by $-\boldsymbol{a}_i$ for any $i\in[m]:=\{1,\dotsc,m\}$. 
Hence, to characterize all optimal frames in $\mathbb{R}^2$, it is enough to consider the case when each $\boldsymbol{a}_i\in\mathcal{T}$.

\begin{theorem}
	\label{struct}
	Assume that $m\geq 3$ is an integer with an odd factor. Let $\mathcal{T}_m$ denote the set of all optimal frames with $m$ vectors in $\mathcal{T}$, where $\mathcal{T}$ is defined in \eqref{def-M-set}.  
	Set $\nv_j:=(\cos\frac{j\pi}{2m}, \sin\frac{j\pi}{2m})^T$ for $j \in [m]$.
	Then we have
	
	\begin{enumerate}
		\item [\rm(i)] 
		Assume that ${\mathcal{A}}= \{\boldsymbol{a}_1,\dotsc,\boldsymbol{a}_{m}\} \in  \mathcal{T}_m$. 
		Then  $\mathcal{A}$ forms a tight frame, and all vectors in ${\mathcal{A}}$ have the same norm, i.e.,
		\[
		\|\boldsymbol{a}_1\|_2=\|\boldsymbol{a}_2\|_2=\cdots=\|\boldsymbol{a}_m\|_2.
		\]
		
		\item[\rm(ii)]
		
		If $\boldsymbol{\varepsilon}=(\varepsilon_1,\dotsc,\varepsilon_{m})^T\in\{\pm1\}^m$ satisfies 
		\begin{equation}\label{eq:partition-inThm}
			\sum_{j=1}^{m} \varepsilon_j \zeta_m^j=0,
		\end{equation}
		where $\zeta_m := e^{\ii\pi/m}$, then the frame 
		${\mathcal A}=\{\nv_j:j\in I\} \cup \{\nv_j^{\perp} : j\in [m]\setminus I\}$ is an optimal frame, i.e., $\mathcal{A} \in \mathcal{T}_m $, where $I=\{j\in[m]:\varepsilon_j=1\}$.
		
		\item[\rm(iii)] 
		For any optimal frame ${\mathcal A}\in {\mathcal T}_m$, there exists a vector $\boldsymbol{\varepsilon}=(\varepsilon_1,\dotsc,\varepsilon_{m})^T\in\{\pm1\}^m$ satisfying \eqref{eq:partition-inThm}, such that ${\mathcal A}=\{\nv_j:j\in I\} \cup \{\nv_j^{\perp} : j\in [m]\setminus I\}$ where $I=\{j\in[m]:\varepsilon_j=1\}$, after a proper scaling and rotation.
		

	\end{enumerate}
\end{theorem}

\begin{remark}
	\label{algorithm}
	Let $m \geq 3$ be an integer with an odd factor. 
	A brute-force algorithm can be used to generate all optimal frames in $\mathcal{T}_m$, up to scaling, rotation, and reflection \footnote{The MATLAB and Maple implementations of the algorithm are available at \url{https://github.com/zxynhy/Optimal-frames-for-the-stability-of-phase-retrieval.git}.}. The algorithm begins by determining all vectors $\boldsymbol{\varepsilon}\in\{\pm1\}^m$ that satisfy \eqref{eq:partition-inThm}. 
	Then Theorem \ref{struct} is applied. 
	Take $m=12$ as an example. 
	After identifying all vectors $\boldsymbol{\varepsilon}\in\{\pm1\}^{12}$ for which \eqref{eq:partition-inThm} holds, by Theorem \ref{struct}, we obtain all the optimal frames in $\mathcal{T}_{12}$ (up to scaling, rotation and reflection), shown in Figure \ref{fig:dodecagons} {\rmfamily\slshape (b1)} and {\rmfamily\slshape (b2)}. 
	Denote by $\#\mathcal{T}_m$ the number of distinct frames in $\mathcal{T}_m$, modulo scaling, rotation and reflection.
	{\rm Table \ref{table1}} lists the number  $\# \mathcal{T}_m$  for integers $m\in [3,15]$ with an odd factor. 
	A promising direction for future research is to design an efficient algorithm capable of finding all vectors $\boldsymbol{\varepsilon}\in\{\pm1\}^m$ that fulfill \eqref{eq:partition-inThm}.
\end{remark}

	\begin{table}[htb]
		\centering
		\caption{
			Number  $\#\mathcal{T}_m$ of inequivalent optimal frames (up to scaling, rotation, and reflection) for integers $3 \leq m \leq 15$ with an odd prime factor.
		}
		\begin{tabular}{c|c|c|c|c|c|c|c|c|c|c|c|c}\label{table1}
			$m$ & 3 & 5 & 6 & 7 & 9 & 10 & 11 & 12 & 13 & 14 & 15 &  \\ \hline
			$\#\mathcal{T}_m$ & 1 & 1 & 1 & 1 & 2 & 1 & 1 & 2 & 1 & 1 & 5 &
		\end{tabular}
	\end{table}


Theorem \ref{struct} establishes that any optimal frame ${\mathcal{A}}= \{\boldsymbol{a}_1,\dotsc,\boldsymbol{a}_{m}\} \subset {\mathbb R}^2$ forms a tight frame when $m\geq 3$ has an odd factor. Inspired by this result, we propose the following conjecture.

\begin{conjecture}
	Any optimal frame ${\mathcal{A}}= \{\boldsymbol{a}_1,\dotsc,\boldsymbol{a}_{m}\} \subset {\mathbb R}^d$ for phase retrieval in ${\mathbb R}^d$ must be a tight frame, provided $m\geq 2d-1$.
\end{conjecture}


\section{Preliminaries}

\subsection{Notation}

For a positive integer $m$, we define $[m] := \{1,\dotsc,m\}$.
We use $\| \boldsymbol{x}\|_2$ to denote the Euclidean norm of a vector $\boldsymbol{x}\in\mathbb{H}^d$, and we use $\| \boldsymbol{A}\|_2$ to denote the spectral norm of a matrix $\boldsymbol{A}\in\mathbb{H}^{m\times d}$. 
We denote the unit sphere in $\mathbb{H}^d$ by $\mathbb{S}_{\mathbb{H}}^{d-1}$, i.e., 
\begin{equation*}
	\mathbb{S}_{\mathbb{H}}^{d-1}=\{\boldsymbol{x}\in \mathbb{H}^d: \| \boldsymbol{x}\|_2=1\}.
\end{equation*} 
If $\mathbb{H}=\mathbb{R}$ then we simply write $\mathbb{S}_{\mathbb{R}}^{d-1}$ as $\mathbb{S}^{d-1}$.



\subsection{Characterizations of Lipschitz constants and condition numbers}

We briefly introduce the
existing results on the optimal lower and upper Lipschitz constants.
\begin{theorem}{\rm\cite{bandeira2014saving,balan2015invertibility,alaifari2017phase}}
	\label{UAH}
	Let $\boldsymbol{A} \in \mathbb{H}^{m \times d}$, where $\mathbb{H}=\mathbb{R}$ or $\mathbb{C}$. Then $U_{\boldsymbol{A}} = \|\boldsymbol{A}\|_2$.
\end{theorem}

\begin{theorem}{\rm\cite{alharbi2024locality}}
	\label{LAH}
	Let $\boldsymbol{A} \in \mathbb{H}^{m \times d}$, where $\mathbb{H}=\mathbb{R}$ or $\mathbb{C}$. Then 
	$$ L_{\boldsymbol{A}} 
	= \min_{\substack{\boldsymbol{x}, \boldsymbol{y} \in \mathbb{H}^d \\ \|\boldsymbol{x}\|_2=1,\|\boldsymbol{y}\|_2\leq 1,\left<\boldsymbol{x},\boldsymbol{y}\right>=0}} 
	\frac{\||\boldsymbol{A} \boldsymbol{x}|-|\boldsymbol{A} \boldsymbol{y}|\|_2}{\operatorname{dist}(\boldsymbol{x}, \boldsymbol{y})}.
	$$
\end{theorem}

The following lemma is useful for proving Theorem \ref{struct}, which is essentially given in \cite{xia2024stability} (see equation (2.2) and (3.16) in \cite{xia2024stability}).

\begin{lemma}
	\label{lowerBound}
	{\rm\cite{xia2024stability}}
	Let $\boldsymbol{A} = [\boldsymbol{a}_1,\ldots, \boldsymbol{a}_m]^T\in \mathbb{R}^{m \times 2}$. Then 
	$$ (L_{\boldsymbol{A}})^2 \,\,\leq\,\, \frac{1}{2}\sum_{j=1}^{m}\|\boldsymbol{a}_j\|_2^2 - \frac{1}{2m \cdot \sin \frac{\pi}{2 m}}\sum_{j=1}^{m}\|\boldsymbol{a}_j\|_2^2.
	$$
\end{lemma}

Inspired by the proof of Theorem 3.3 in \cite{xia2024stability}, we present the following lemma, which provides an alternative expression for the condition number of tight frames.

\begin{lemma}
	\label{eqform}
	Assume that $\{\boldsymbol{a}_1,\dotsc,\boldsymbol{a}_{m}\} \subset  \mathbb{H}^{d}$ forms a tight frame, where $\mathbb{H}=\mathbb{R}$ or $\mathbb{C}$. 
	Then  
	\begin{equation}
		\label{beta_equi}
		\beta_{\boldsymbol{A}} =
		\left(1-\frac{1}{C} 
		\underset{\substack{\boldsymbol{x},\boldsymbol{y} \in \mathbb{H}^d, \\ \|\boldsymbol{x}\|_2=\|\boldsymbol{y}\|_2=1,  \langle \boldsymbol{x},\boldsymbol{y} \rangle=0}}{\max}
		\sum_{j=1}^m \left|\boldsymbol{x}^* \boldsymbol{a}_j \boldsymbol{a}_j^* \boldsymbol{y}\right|\right)^{-\frac{1}{2}},
	\end{equation}
	where $\boldsymbol{A} = [\boldsymbol{a}_1,\dotsc,\boldsymbol{a}_{m}]^* \in  \mathbb{H}^{m\times d}$ and  $C = \frac{1}{d}\sum_{j=1}^{m}{\|\boldsymbol{a}_j\|_2^2}$.
\end{lemma}
\begin{proof}
	Since $\{\boldsymbol{a}_1,\dotsc,\boldsymbol{a}_{m}\} \subset  \mathbb{H}^{d}$ forms a tight frame, we have $\boldsymbol{A}^{*}\boldsymbol{A} = C\cdot  \boldsymbol{I}_d$, where  $C=\frac{1}{d}\sum_{j=1}^{m}{\|\boldsymbol{a}_j\|_2^2}$.
	By Theorem \ref{LAH} we have 
	\begin{equation}
		\label{LALem}
		(L_{\boldsymbol{A}})^2 = \min_{\substack{\boldsymbol{x}, \boldsymbol{y} \in \mathbb{H}^d \\ \|\boldsymbol{x}\|_2=1,\|\boldsymbol{y}\|_2\leq 1,\left<\boldsymbol{x},\boldsymbol{y}\right>=0}} 
		\frac{\||\boldsymbol{A} \boldsymbol{x}|-|\boldsymbol{A} \boldsymbol{y}|\|_2^2}{\operatorname{dist}^2(\boldsymbol{x}, \boldsymbol{y})}.
	\end{equation}
	If $\boldsymbol{y} = \mathbf{0}$, then we simply have
	\begin{equation}\label{eq:xu4}
		\frac{\||\boldsymbol{A} \boldsymbol{x}| - |\boldsymbol{A} \boldsymbol{y}|\|_2^2}{\operatorname{dist}^2(\boldsymbol{x}, \boldsymbol{y})} = \frac{\|\boldsymbol{A} \boldsymbol{x}\|_2^2}{\|\boldsymbol{x}\|_2^2} = C.
	\end{equation}
	We next consider the case when $\boldsymbol{y} \neq \mathbf{0}$. Assume that $\boldsymbol{x},\, \boldsymbol{y} \in \mathbb{H}^d$, $\|\boldsymbol{x}\|_2=1,\,0<\|\boldsymbol{y}\|_2\leq 1$, and $\left<\boldsymbol{x},\boldsymbol{y}\right>=0$. 
	A direct calculation shows $
	\operatorname{dist}^2(\boldsymbol{x}, \boldsymbol{y}) = \|\boldsymbol{x} - \boldsymbol{y}\|_2^2 = 1 + \|\boldsymbol{y}\|_2^2$
	and
	\begin{align*}
		\left\||\boldsymbol{A} \boldsymbol{x}| - |\boldsymbol{A} \boldsymbol{y}|\right\|_2^2 & = \boldsymbol{x}^* \boldsymbol{A}^* \boldsymbol{A} \boldsymbol{x} + \boldsymbol{y}^* \boldsymbol{A}^* \boldsymbol{A} \boldsymbol{y} - 2 \sum_{j=1}^m \left|\boldsymbol{x}^* \boldsymbol{a}_j \boldsymbol{a}_j^* \boldsymbol{y}\right| \\ & = C\left(1 + \|\boldsymbol{y}\|_2^2\right) - 2\|\boldsymbol{y}\|_2 \sum_{j=1}^m \left|\boldsymbol{x}^* \boldsymbol{a}_j \boldsymbol{a}_j^* \widetilde{\boldsymbol{y}}\right|,
	\end{align*}
	where $\widetilde{\boldsymbol{y}}:=\boldsymbol{y}/\|\boldsymbol{y}\|_2$.
	It follows that
	\begin{equation}
		\label{C-sum}
		\frac{\||\boldsymbol{A} \boldsymbol{x}| - |\boldsymbol{A} \boldsymbol{y}|\|_2^2}{\operatorname{dist}^2(\boldsymbol{x}, \boldsymbol{y})} 
		= C - \frac{2\sum_{j=1}^m \left|\boldsymbol{x}^* \boldsymbol{a}_j \boldsymbol{a}_j^* \widetilde{\boldsymbol{y}}\right|}{\frac{1}{\|\boldsymbol{y}\|_2}+\|\boldsymbol{y}\|_2} 
		\overset{(a)}{\geq} C - \sum_{j=1}^m \left|\boldsymbol{x}^* \boldsymbol{a}_j \boldsymbol{a}_j^* \widetilde{\boldsymbol{y}}\right|.
	\end{equation}
	Here, we apply the Cauchy-Schwarz inequality in ($a$), where the equality is achieved if and only if $\|\boldsymbol{y}\|_2=1$. 
	Therefore, combining \eqref{LALem} with \eqref{eq:xu4} and \eqref{C-sum}, we have
	\begin{equation}\label{eq:LA2}
		\begin{aligned}
			(L_{\boldsymbol{A}} )^2
			& = \underset{\substack{\boldsymbol{x},\boldsymbol{y} \in \mathbb{H}^d, \\ \|\boldsymbol{x}\|_2=\|\boldsymbol{y}\|_2 = 1, \, \langle \boldsymbol{x},\boldsymbol{y} \rangle=0}}{\min}\left(C - \sum_{j=1}^m \left|\boldsymbol{x}^* \boldsymbol{a}_j \boldsymbol{a}_j^* \boldsymbol{y}\right|\right) \\ 
			& = C - \underset{\substack{\boldsymbol{x},\boldsymbol{y} \in \mathbb{H}^d, \\ \|\boldsymbol{x}\|_2=\|\boldsymbol{y}\|_2=1, \, \langle \boldsymbol{x},\boldsymbol{y} \rangle=0}}{\max}\sum_{j=1}^m \left|\boldsymbol{x}^* \boldsymbol{a}_j \boldsymbol{a}_j^* \boldsymbol{y}\right|.
		\end{aligned}
	\end{equation}
	Combining  \eqref{eq:LA2} with $\beta_{\boldsymbol{A}} = U_{\boldsymbol{A}} / L_{\boldsymbol{A}}$ and $U_{\boldsymbol{A}} =\|\boldsymbol{A}\|_2= \sqrt{C}
	$, we arrive at \eqref{beta_equi}. 
\end{proof}

The following lemma establishes that, for determining the minimal condition number of $\boldsymbol{A} \in \mathbb{H}^{m \times d}$, it suffices to consider cases where its rows form a tight frame.

\begin{lemma}
	\label{optistight2}
	Assume that $m\geq d$. We have
	\begin{equation*}
		\min _{\boldsymbol{A} \in \mathbb{H}^{m \times d}} 
		\beta_{\boldsymbol{A}} 
		=\min _{\boldsymbol{A} \in \mathbb{H}^{m \times d}, \boldsymbol{A}^*\boldsymbol{A}=\boldsymbol{I}_d} 
		\beta_{\boldsymbol{A}} .
	\end{equation*}
\end{lemma}

\begin{proof}
	If  ${\rm rank}({\boldsymbol{A}})< d$, then  $\beta_{\boldsymbol{A}}=+\infty$.
	We next consider the case when  ${\boldsymbol{A}}\in  \mathbb{H}^{m\times d}$ has rank $d$.
	Let $\boldsymbol{M}=\boldsymbol{A}^*\boldsymbol{A}\in  \mathbb{H}^{d\times d}$ and $\boldsymbol{B}=\boldsymbol{A} \boldsymbol{M}^{-\frac{1}{2}}\in  \mathbb{H}^{m\times d}$.
	We claim that
	\begin{equation}\label{eq:claim1}
		\beta_{{\boldsymbol{A}}} \geq \beta_{{\boldsymbol{B}}} .
	\end{equation}
	Note that the rows of ${\boldsymbol{B}}$ form a tight frame.
	Also note that the condition number $\beta_{\boldsymbol{B}}$ is invariant under the scaling of ${\boldsymbol{B}}$.
	Thus, \eqref{eq:claim1} implies the conclusion. 
	It suffices  to prove \eqref{eq:claim1}.
	
	Note that for any $\boldsymbol{x},\boldsymbol{y}\in\mathbb{H}^d$, 
	\begin{equation*}
		\begin{aligned}
			\||\boldsymbol{B} \boldsymbol{x}|-|\boldsymbol{B} \boldsymbol{y}|\|_2
			&=\||\boldsymbol{A}\boldsymbol{M}^{-\frac{1}{2}} \boldsymbol{x}|-|\boldsymbol{A}\boldsymbol{M}^{-\frac{1}{2}} \boldsymbol{y}|\|_2	
			=\||\boldsymbol{A} \widehat{\boldsymbol{x}}|-|\boldsymbol{A}\widehat{\boldsymbol{y}}|\|_2,
		\end{aligned}	
	\end{equation*}
	where $\widehat{\boldsymbol{x}}=\boldsymbol{M}^{-\frac{1}{2}} \boldsymbol{x}$ and $\widehat{\boldsymbol{y}}=\boldsymbol{M}^{-\frac{1}{2}} \boldsymbol{y}$.
	Also note that for any $c \in \mathbb{H}$,
	\begin{equation*}
		\|\boldsymbol{x}-c \cdot \boldsymbol{y}\|_2
		=\|\boldsymbol{M}^{\frac{1}{2}} (\widehat{\boldsymbol{x}}-c \cdot  \widehat{\boldsymbol{y}})\|_2
		\leq \|\boldsymbol{M}^{\frac{1}{2}}  \|_2\cdot 
		\|\widehat{\boldsymbol{x}}-c \cdot  \widehat{\boldsymbol{y}}\|_2	,
	\end{equation*}
	which implies that
	\begin{equation*}
			\operatorname{dist}_{\mathbb{H}}(\boldsymbol{x}, \boldsymbol{y})
			=\min_{c \in \mathbb{H},|c|=1}  \|\boldsymbol{x}-c \cdot \boldsymbol{y}\|_2 \\
			\leq \|\boldsymbol{M}^{\frac{1}{2}}  \|_2\cdot \min_{c \in \mathbb{H},|c|=1}  \|\widehat{\boldsymbol{x}}-c \cdot \widehat{\boldsymbol{y}}\|_2
			=\|\boldsymbol{M}^{\frac{1}{2}}  \|_2\cdot  \operatorname{dist}_{\mathbb{H}}(\widehat{\boldsymbol{x}},\widehat{ \boldsymbol{y}}).
	\end{equation*}
	Therefore,
	\begin{equation*}
		\frac{\||\boldsymbol{B} \boldsymbol{x}|-|\boldsymbol{B} \boldsymbol{y}|\|_2}
		{\operatorname{dist}_{\mathbb{H}}(\boldsymbol{x}, \boldsymbol{y})}
		=\frac{\||\boldsymbol{A} \widehat{\boldsymbol{x}}|-|\boldsymbol{A}\widehat{\boldsymbol{y}}|\|_2}
		{\operatorname{dist}_{\mathbb{H}}(\boldsymbol{x}, \boldsymbol{y})}
		\geq \frac{1}{\|\boldsymbol{M}^{\frac{1}{2}}  \|_2}\cdot  \frac{\||\boldsymbol{A} \widehat{\boldsymbol{x}}|-|\boldsymbol{A}\widehat{\boldsymbol{y}}|\|_2}
		{  \operatorname{dist}_{\mathbb{H}}(\widehat{\boldsymbol{x}},\widehat{ \boldsymbol{y}})}.
	\end{equation*}
	Since $\boldsymbol{M}$ has full rank, we have
	\begin{equation*}
			L_{\boldsymbol{B}}  
			=\inf _{\substack{\boldsymbol{x}, \boldsymbol{y} \in \mathbb{H}^d \\ \operatorname{dist}_{\mathbb{H}}(\boldsymbol{x}, \boldsymbol{y}) \neq 0}}
			\frac{\||\boldsymbol{B} \boldsymbol{x}|-|\boldsymbol{B} \boldsymbol{y}|\|_2}
			{\operatorname{dist}_{\mathbb{H}}(\boldsymbol{x}, \boldsymbol{y})} \\
			\geq \frac{1}{\|\boldsymbol{M}^{\frac{1}{2}}  \|_2}\cdot 
			\inf _{\substack{\widehat{\boldsymbol{x}}, \widehat{\boldsymbol{y}} \in \mathbb{H}^d \\ \operatorname{dist}_{\mathbb{H}}(\widehat{\boldsymbol{x}}, \widehat{\boldsymbol{y}}) \neq 0}}
			\frac{\||\boldsymbol{A} \widehat{\boldsymbol{x}}|-|\boldsymbol{A}\widehat{\boldsymbol{y}}|\|_2}
			{  \operatorname{dist}_{\mathbb{H}}(\widehat{\boldsymbol{x}},\widehat{ \boldsymbol{y}})}
			=\frac{L_{\boldsymbol{A}} }{\|\boldsymbol{M}^{\frac{1}{2}}  \|_2}.
	\end{equation*}
	Note that ${\boldsymbol{B}}^*{\boldsymbol{B}}=\boldsymbol{M}^{-\frac{1}{2}} \boldsymbol{A}^*\boldsymbol{A} \boldsymbol{M}^{-\frac{1}{2}}=\boldsymbol{I}_d$ and $\|\boldsymbol{M}^{\frac{1}{2}}\|_2=\|\boldsymbol{A}\|_2$, so we have $U_{\boldsymbol{B}} =\|\boldsymbol{B}\|_2=1$ and
	\begin{equation*}
		\beta_{{\boldsymbol{B}}} 
		=\frac{U _{\boldsymbol{B}}}{L _{\boldsymbol{B}}}
		\leq \frac{\|\boldsymbol{M}^{\frac{1}{2}}\|_2 }{L _{\boldsymbol{A}}}
		=\frac{\|\boldsymbol{A}\|_2 }{L _{\boldsymbol{A}}}
		=\frac{U _{\boldsymbol{A}} }{L _{\boldsymbol{A}}}
		=\beta_{{\boldsymbol{A}}} .
	\end{equation*}
	We arrive at our conclusion.
\end{proof}

\subsection{The perimeter-maximizing isodiametric problem}




To prove our main stability results for phase retrieval, we introduce the perimeter-maximizing isodiametric problem, first studied by Reinhardt (1922) and later developed in a substantial literature (e.g., \cite{reinhardt1922extremale,vincze1950geometrical,datta1997discrete,gashkov1985inequalities,gashkov2007inequalities,mossinghoff2011enumerating,bingane2021maximal}).
A convex polygon in $\mathbb{R}^2$ with $m$ edges is called a convex $m$-gon. 
The diameter of a polygon $P$, denoted by $\operatorname{diam}( {P})$, is the largest distance between any pair of its vertices, i.e., $\operatorname{diam}( {P}) :=  \max_{\boldsymbol{x},\boldsymbol{y} \,\in  {P}}{\|\boldsymbol{x}-\boldsymbol{y}\|_2}$. 
The perimeter of a polygon $P$, denoted by $\operatorname{perim}( {P})$, is the sum of the length of its edges.
Denote $r({P})$ as the diameter-to-perimeter ratio of a polygon  $ {P} $, i.e., 
\begin{equation}\label{eq:rp}
	r(P):=\frac{\mathrm{diam}(P)}{\mathrm{perim}(P)}.
\end{equation}
The perimeter-maximizing isodiametric problem asks: 
\begin{problem}\label{pr1}
	Among all convex $m$-gons with fixed diameter, which ones have the maximal perimeter? 
	Equivalently, which convex $m$-gons minimize the diameter-to-perimeter ratio $r(P)$?	
\end{problem}

For clarity, we introduce the following definitions regarding polygons, which are used consistently throughout this paper.
\begin{definition}\label{de:polygon}
	\begin{enumerate}
		\item[\rm (i)] Optimal polygon: A convex $m$-gon is defined as optimal if it is a solution to Problem \ref{pr1}.
		We say that the solution to Problem \ref{pr1} is unique if it is unique up to scaling, translations, rotations, reflections, or any combination of these transformations.
		\item[\rm (ii)] Equilateral polygon: A polygon is called equilateral if all its edges have equal length.
		\item[\rm (iii)] Regular $m$-gon: A regular $m$-gon, denoted as $P_m$, is an equilateral polygon with equal interior angles.
		\item[\rm (iv)] Strictly convex polygon: A convex $m$-gon is called strictly convex if all its interior angles are strictly less than $\pi$.
		We use $\mathcal{P}_m$ to denote the set of all strictly convex $m$-gons in $\mathbb{R}^2$.
		\item[\rm (v)] Edge vector: Assume the vertices of a convex $m$-gon are ordered counterclockwise. For each edge, the edge vector is defined by subtracting the coordinates of its starting vertex from those of its ending vertex. 
		The edge set of the polygon is the collection of all such edge vectors.
	\end{enumerate}
\end{definition}


It is known that the optimal convex $m$-gon must be strictly convex  \cite{reinhardt1922extremale}. 
Theorem \ref{knownresults} summarizes the existing results for the perimeter-maximizing isodiametric problem.

\begin{theorem}\label{knownresults}
	\begin{enumerate}
		\item[\rm (i)]  {\rm \cite{reinhardt1922extremale, datta1997discrete, gashkov2007inequalities,gashkov1985inequalities,vincze1950geometrical,mossinghoff2011enumerating}} Let $m \geq 3$ be a positive integer that has an odd factor. 
		Then 
		\begin{equation*}
			\min_{ {P} \in \mathcal{P}_m} r( {P})=
			\frac{1}{2m \cdot \sin \frac{\pi}{2 m}},
		\end{equation*}
		where equality is achieved by finitely many equilateral strictly convex $m$-gons. 
		
		\item[\rm (ii)]  {\rm\cite{reinhardt1922extremale, datta1997discrete, gashkov2007inequalities,gashkov1985inequalities}} Let $m = 2^s$, where $s$ is an integer and $s \geq 2$.
		Then  
		\begin{equation*}
			\min_{ {P} \in \mathcal{P}_m} r( {P}) > \frac{1}{2m \cdot \sin \frac{\pi}{2 m}}.
		\end{equation*}
		
		\item[\rm (iii)] {\rm\cite{reinhardt1922extremale, datta1997discrete, gashkov2007inequalities,gashkov1985inequalities,vincze1950geometrical}} The regular $m$-gon  $ P_m$ is optimal if and only if $m$ is odd. Furthermore, $ {P}_m$ is uniquely optimal if and only if $m$ is a prime. 
	\end{enumerate}
\end{theorem}

When $m \geq 3$ has an odd factor, the optimal convex $m$-gons can be nicely characterized \cite{reinhardt1922extremale, vincze1950geometrical, datta1997discrete, gashkov1985inequalities, gashkov2007inequalities, mossinghoff2011enumerating}.
When $m\geq 3$ has no odd factors, i.e., $m=2^s, s\geq 2$, the problem remains largely unresolved. The characterization is known only for $m=4$ \cite{mossinghoff20061,taylor1953some,tamvakis1987perimeter} and $m=8$ \cite{griffiths1975pi,audet2007small,bingane2021maximal}.

In the following we introduce several basic properties of convex polygons, which will be utilized in later analysis.


\begin{lemma}{\rm\cite{schneider2013convex}}
	\label{diam_suppf}
	For any convex polygon $ {P} \subset \mathbb{R}^2$, we have 
	\[
	\operatorname{diam}( {P}) = \max_{\boldsymbol{u} \in \mathbb{S}^1} 
	h( {P},\boldsymbol{u})+h( {P},-\boldsymbol{u}),
	\]
	where $h( {P},\boldsymbol{z}) := \max_{\boldsymbol{x} \in  {P}}\langle\boldsymbol{x},\boldsymbol{z} \rangle$ for any $\boldsymbol{z}\in \mathbb{S}^1$.
\end{lemma}

\begin{lemma}
	\label{tight_gon-form}
	Assume that $E=\{\boldsymbol{e}_1,\dotsc,\boldsymbol{e}_{m}\}\subset \mathbb{R}^2$ is the edge set of a convex $m$-gon $ {P} \subset \mathbb{R}^2$.
	Then 
	\begin{equation}
		\label{R}
		r( {P}) = \frac{\underset{\boldsymbol{u} \in \mathbb{S}^1}{\max}\,\sum_{j=1}^{m} {|\langle \boldsymbol{e}_{j},\boldsymbol{u} \rangle|}}{2\sum_{j=1}^{m}\|\boldsymbol{e}_j\|_2}.
	\end{equation}
	
\end{lemma}

\begin{proof}
	We first consider the case when $P$ is strongly convex.
	Since $r(P)=\frac{\mathrm{diam}(P)}{\mathrm{perim}(P)}$ and the perimeter $\mathrm{perim}(P)$ equals $\sum_{j=1}^{m}\|\boldsymbol{e}_j\|_2$, 
	it is enough to prove 
	\begin{equation}
		\label{forlatter}
		\operatorname{diam} ( P)= \frac{1}{2}\max_{\boldsymbol{u} \in \mathbb{S}^1}\,\sum_{j=1}^{m} {|\langle \boldsymbol{e}_{j},\boldsymbol{u} \rangle|}.
	\end{equation}
	
	Without loss of generality, we assume that $\boldsymbol{e}_1,\dotsc,\boldsymbol{e}_{m}$ are listed in counterclockwise order. 
	Let $\boldsymbol{v}_1,\dotsc,\boldsymbol{v}_{m}$ be the vertices of $ P$ such that $\boldsymbol{e}_j = \boldsymbol{v}_{j+1} - \boldsymbol{v}_{j}$ for each $j \in [m]$, where we denote $\boldsymbol{v}_{m+1} = \boldsymbol{v}_{1}$. For any $\boldsymbol{u} \in \mathbb{S}^1$, since $ P$ is convex, there exist ${j_1},\,j_2 \in [m]$ depending on $\boldsymbol{u}$ such that 
	\begin{equation}
		\label{x1x2}
		\boldsymbol{v}_{j_1} \in \operatorname*{argmin}_{\boldsymbol{x} \in  P}\langle\boldsymbol{x},\boldsymbol{u} \rangle
		\quad\text{and}\quad
		\boldsymbol{v}_{j_2} \in  \operatorname*{argmax}_{\boldsymbol{x} \in  P}\langle\boldsymbol{x},\boldsymbol{u} \rangle.
	\end{equation}
	Assume without loss of generality that ${j_1}<j_2$. 
	By Lemma \ref{diam_suppf}, we have
	\begin{equation}
		\label{diam1}
		\begin{aligned}
			\operatorname{diam} ( P) 
			& = \max_{\boldsymbol{u} \in \mathbb{S}^1} h( P,\boldsymbol{u})+h( P,-\boldsymbol{u})   = \max_{\boldsymbol{u} \in \mathbb{S}^1} (\max_{\boldsymbol{x} \in  {P}}\langle\boldsymbol{x},\boldsymbol{u} \rangle+\max_{\boldsymbol{x} \in  {P}}\langle\boldsymbol{x},-\boldsymbol{u} \rangle) \\ 
			& = \max_{\boldsymbol{u} \in \mathbb{S}^1} (\max_{\boldsymbol{x} \in  {P}}\langle\boldsymbol{x},\boldsymbol{u} \rangle-\min_{\boldsymbol{x} \in  {P}}\langle\boldsymbol{x},\boldsymbol{u} \rangle)  = \max_{\boldsymbol{u} \in \mathbb{S}^1} \, { \langle\boldsymbol{v}_{j_2}-\boldsymbol{v}_{j_1},\boldsymbol{u} \rangle }.
		\end{aligned}
	\end{equation}
	We claim that
	\begin{equation}
		\label{claim}
		\langle\boldsymbol{e}_{j},\boldsymbol{u} \rangle
		\begin{cases}
			\geq 0 & \text{if $j_1 \leq j \leq j_2-1$,}\\
			\leq 0 & \text{else.}	
		\end{cases}
	\end{equation}
	Then we have
	\begin{equation}
		\label{d1}
			\langle \boldsymbol{v}_{j_2}-\boldsymbol{v}_{j_1},\boldsymbol{u} \rangle 
			= \langle  \sum_{j=j_1}^{j_2-1}{\boldsymbol{e}_{j}},\boldsymbol{u} \rangle 
			\overset{(a)}= \langle \frac{1}{2}\sum_{j=j_1}^{j_2-1}{\boldsymbol{e}_{j}} - \frac{1}{2}\sum_{j=1}^{j_1-1}{\boldsymbol{e}_{j}}-\frac{1}{2}\sum_{j=j_2}^{m}{\boldsymbol{e}_{j}},\boldsymbol{u} \rangle \\
			\overset{(b)}=\frac{1}{2}\sum_{j=1}^{m} {|\langle \boldsymbol{e}_{j},\boldsymbol{u} \rangle|},
	\end{equation}
	where $(a)$ follows from $\sum_{j=1}^{m}\boldsymbol{e}_{j}=\boldsymbol{0}$ and $(b)$ follows from \eqref{claim}.
	Substituting \eqref{d1} into \eqref{diam1}, we arrive at \eqref{forlatter}.

	It remains to prove \eqref{claim}. 
	We first prove $\langle\boldsymbol{e}_{j},\boldsymbol{u} \rangle\geq 0$ for each $j_1\leq j\leq j_2-1$. 
	By the definition of $j_1$ and $j_2$, we have 
	\begin{equation} \label{tiny1}
		\begin{aligned}
			\langle\boldsymbol{e}_{j_1},\boldsymbol{u} \rangle
			=\langle\boldsymbol{v}_{j_1+1},\boldsymbol{u} \rangle - \langle\boldsymbol{v}_{j_1},\boldsymbol{u} \rangle
			\geq 0
			\quad\text{and} \quad
			\langle\boldsymbol{e}_{j_2-1},\boldsymbol{u} \rangle
			=\langle\boldsymbol{v}_{j_2},\boldsymbol{u} \rangle -\langle\boldsymbol{v}_{j_2-1},\boldsymbol{u} \rangle \geq 0. 
		\end{aligned}
	\end{equation}
	Suppose, for contradiction, that there exists an integer $s$ with $ j_1+1\leq s\leq j_2-2$ such that $\langle\boldsymbol{e}_{s},\boldsymbol{u} \rangle< 0$. 
	Then there must exist an integer $t \in \{s+1,\dotsc,j_2-1\}$ satisfying 
	\begin{equation}
		\label{vt}
		\langle \boldsymbol{e}_{t-1},\boldsymbol{u} \rangle < 0,
		\quad\text{and} \quad
		\langle \boldsymbol{e}_{t},\boldsymbol{u} \rangle \geq 0,
	\end{equation}
	since otherwise we have $\langle \boldsymbol{e}_{i},\boldsymbol{u} \rangle < 0$ for each $s+1\leq i\leq j_2-1$, which contradicts with \eqref{tiny1}.
	Since $P$ is strictly convex, for any $\boldsymbol{x} \in P$ there exist $\gamma_1,\gamma_2\geq 0$ depending on $\boldsymbol{x}$ such that
	\[
	\boldsymbol{x} - \boldsymbol{v}_{t} = \gamma_1\cdot (-\boldsymbol{e}_{t-1})+\gamma_2\cdot \boldsymbol{e}_{t}.
	\]
	Taking inner products with $\boldsymbol{u}$ and using \eqref{vt}, we have
	\[
	\langle \boldsymbol{x},\boldsymbol{u} \rangle - \langle \boldsymbol{v}_{t},\boldsymbol{u} \rangle 
	= \gamma_1\langle -\boldsymbol{e}_{t-1},\boldsymbol{u} \rangle+\gamma_2\langle \boldsymbol{e}_{t},\boldsymbol{u} \rangle \geq 0, \quad\forall \boldsymbol{x} \in P,
	\]
	implying that
	\begin{equation}
		\label{vt_mini}
		\langle \boldsymbol{v}_{j_1},\boldsymbol{u} \rangle=\langle \boldsymbol{v}_{t},\boldsymbol{u} \rangle = {\min}_{\boldsymbol{x} \in  P}\langle\boldsymbol{x},\boldsymbol{u} \rangle.
	\end{equation}
	It is well known that a linear function on a convex polygon attains its extrema at either a vertex or along an edge of the polygon \cite{schrijver1998theory}.
	Hence, we must have $t=j_1+1$.
	This contradicts with $t\geq s+1\geq j_1+2$.
	Therefore, we have $\langle\boldsymbol{e}_{j},\boldsymbol{u} \rangle\geq 0$ for each $j_1\leq j\leq j_2-1$.
	Similarly, we can prove that $\langle\boldsymbol{e}_{j},\boldsymbol{u} \rangle\leq 0$ for each $j<j_1$ and $j\geq j_2$. Hence, we arrive at \eqref{claim}.
	
	We next consider the case when $P$ is not strictly convex.
	By replacing all vectors in the same direction with their sum, we obtain a strictly convex $k$-gon $ \widehat{P}\in\mathcal{P}_k$ with the edge set  $\widehat{E}=\{\widehat{\boldsymbol{e}}_1,\dotsc,\widehat{\boldsymbol{e}}_{k}\}$, where $k$ is an integer less than $m$.
	Note that $\sum_{j=1}^{m} \|\boldsymbol{e}_{j}\|_2=\sum_{j=1}^{k} \|\widehat{\boldsymbol{e}}_{j}\|_2$ and $\sum_{j=1}^m |\langle \boldsymbol{e}_j, \boldsymbol{u} \rangle |=\sum_{j=1}^k |\langle \widehat{\boldsymbol{e}}_j, \boldsymbol{u} \rangle |$ for any $\boldsymbol{u}\in\mathbb{S}^1$. 
	Hence, we have
	\begin{equation*}
		r(P)=r(\widehat{P})
		=\frac{\underset{\boldsymbol{u} \in \mathbb{S}^1}{\max}\,\sum_{j=1}^{m} {|\langle \widehat{\boldsymbol{e}}_{j},\boldsymbol{u} \rangle|}}{2\sum_{j=1}^{m}\|\widehat{\boldsymbol{e}}_j\|_2}
		=\frac{\underset{\boldsymbol{u} \in \mathbb{S}^1}{\max}\,\sum_{j=1}^{m} {|\langle \boldsymbol{e}_{j},\boldsymbol{u} \rangle|}}{2\sum_{j=1}^{m}\|\boldsymbol{e}_j\|_2}.
	\end{equation*}
	Hence, we arrive at our conclusion. This completes the proof.
\end{proof}

The following theorem discovered by Minkowski indicates that a zero-sum vector set in $\mathbb{R}^2$ with pairwise distinct directions uniquely determines a strictly convex polygon.
\begin{theorem}{\rm\cite{klain2004minkowski}}
	\label{Minkowski}
	Suppose $\boldsymbol{u}_1,\dotsc,\boldsymbol{u}_m \in \mathbb{R}^2$ are unit vectors that span $\mathbb{R}^2$, and suppose that $\alpha_1,\dotsc,\alpha_m >0$. Then there exists a convex $m$-gon $ {P}$ in $\mathbb{R}^2$, having edge unit outer normals $\boldsymbol{u}_1,\dotsc,\boldsymbol{u}_m$ and corresponding edge lengths $\alpha_1,\dotsc,\alpha_m >0$, if and only if
	$\sum_{j=1}^{m}\alpha_j\boldsymbol{u}_j = \boldsymbol{0}.$
	Moreover, such a polygon $ {P}$ is unique up to translation.
\end{theorem}
{
	The next statement seems to be well known, but we could not find a precise reference. For completeness, we include a proof.
}
\begin{corollary}
	\label{Minkowski2}
	Suppose $E=\{\boldsymbol{e}_1,\dotsc,\boldsymbol{e}_m\} \subset \mathbb{R}^2$.
	Then there exists a strictly convex $m$-gon $ {P}$ in $\mathbb{R}^2$, having the edge set $E$ if and only if $\boldsymbol{e}_1,\dotsc,\boldsymbol{e}_m $ are nonzero vectors with pairwise distinct directions and
	$\sum_{j=1}^{m}\boldsymbol{e}_j = \boldsymbol{0}$.
	Moreover, such a polygon $ {P}$ is unique up to translation.
\end{corollary}
\begin{proof}
	Assume that $\boldsymbol{e}_1,\dotsc,\boldsymbol{e}_m $ are nonzero vectors with pairwise distinct directions and
	$\sum_{j=1}^{m}\boldsymbol{e}_j = \boldsymbol{0}$. For all $j \in [m]$, rotating each vector $\frac{\boldsymbol{e}_j}{\|\boldsymbol{e}_j\|_2}$ counterclockwise by $\frac{\pi}{2}$ yields $\boldsymbol{u}_j$, while preserving the zero-sum property $\sum_{j=1}^{m}\|\boldsymbol{e}_j\|_2\boldsymbol{u}_j=\boldsymbol{0}$. 
	It is easy to verify that $\boldsymbol{u}_1,\dotsc,\boldsymbol{u}_m \in \mathbb{R}^2$ are unit vectors that span $\mathbb{R}^2$. Therefore we can utilize Theorem \ref{Minkowski} 
	to obtain a convex $m$-gon ${P}$ with edge unit outer normals $\boldsymbol{u}_1,\dotsc,\boldsymbol{u}_m$ and corresponding edge lengths $\|\boldsymbol{e}_1\|_2,\dotsc,\|\boldsymbol{e}_m\|_2$, i.e., with edge vectors $\boldsymbol{e}_1,\dotsc,\boldsymbol{e}_m$. 
	Moreover, $P$ is unique, disregarding the translation of polygons. Now we prove that the convex polygon $P$ is strictly convex. Suppose, by contradiction, that $P$ is convex but not strictly convex. Then $P$ must have an interior angle of $\pi$, implying the existence of two edge vectors in the same direction. This contradicts the assumption that $\boldsymbol{e}_1,\dotsc,\boldsymbol{e}_m$ have pairwise distinct directions. Thus, $P$ is strictly convex.
	
	Conversely, assume that $ {P}$ is a strictly convex $m$-gon  with the edge set $E$. 
	Assume that $\boldsymbol{e}_1,\dotsc,\boldsymbol{e}_m$ are arranged counterclockwise.
	It is clear that $\sum_{j=1}^{m}\boldsymbol{e}_j =\boldsymbol{0}$ and each $\boldsymbol{e}_i$ is nonzero. 
	Without loss of generality, it is enough to prove that $\boldsymbol{e}_1$ and $\boldsymbol{e}_i$ have distinct directions for any $i> 1$. 
	For each $i \in \{2,\ldots,m\}$, define $\alpha_i$ as the counterclockwise angle from $\boldsymbol{e}_1$ to $\boldsymbol{e}_i$. 
	Since $P$ is strictly convex, we have $\alpha_i\in(0,2\pi)$, and $\alpha_i$ is strictly increasing as $i$ increases. 
	Thus, there does not exist an integer $i\in\{2,\ldots,m\}$  such that $\alpha_i=0$, i.e., $\boldsymbol{e}_1$ and $\boldsymbol{e}_j$ have the same direction. This completes the proof.
\end{proof}

\section{Relationship between optimal tight frames and optimal polygons: Proof of Theorem \ref{mincond}}\label{S3-relationship}

In this section, we aim to prove Theorem \ref{mincond}. 
We begin by introducing a  bijection from $\mathcal{T}$ to  $\mathbb{R}^2$.

\begin{definition}
	Let $f: \mathcal{T} \to \mathbb{R}^2$ be a map defined by 
	\begin{equation}\label{def-F-func}
		f( \boldsymbol{a} )\,\,=\,\, t^2\cdot (\cos2\phi, \sin2\phi)^T
	\end{equation}
	for any $\boldsymbol{a}= t\cdot (\cos\phi, \sin\phi)^T\in \mathcal{T}$, where $\mathcal{T}$ is defined in \eqref{def-M-set}.
	
\end{definition}

\begin{remark}
	The vector $f(\boldsymbol{a})$ is called the diagram vector associated with $\boldsymbol{a}$ \cite{han2007frames}. We now show that the map $f:\mathcal{T}\to \mathbb{R}^2$ is a bijection.
	First, it holds that $\boldsymbol{a}=\boldsymbol{0}$ if and only if $f(\boldsymbol{a})=\boldsymbol{0}$.  Second, for any nonzero vector $\boldsymbol{e}_0\in \mathbb{R}^2$, it possesses a unique representation in the form $\boldsymbol{e}_0=t_0\cdot (\cos{2\phi_0},\sin{2\phi_0})^{T}$, where $t_0> 0$ and $\phi_0\in(0,\pi]$. Given this, the unique pre-image $\boldsymbol{a}_0\in \mathcal{T}$ such that $f(\boldsymbol{a}_0)=\boldsymbol{e}_0$ is given by $\boldsymbol{a}_0=\sqrt{t_0}\cdot (\cos{\phi_0},\sin{\phi_0})^{T}$.
	Therefore, $f$ is both injective and surjective, implying it is a bijection. This ensures that the inverse map $f^{-1}:\mathbb{R}^2\to \mathcal{T}$ is well-defined, and is explicitly given by $f^{-1}(\boldsymbol{e}_0)=\boldsymbol{a}_0$.
\end{remark}


In the following theorem, we establish a bijective correspondence between optimal tight frames with $m$ vectors in $\mathbb{R}^2$ and optimal convex $m$-gons. This will play a key role in proving Theorem \ref{mincond}. 
We postpone the proof of Theorem \ref{main} to the end of this section. 

\begin{theorem}
	\label{main}
	Assume that $\mathcal{A}=\{\boldsymbol{a}_1,\ldots, \boldsymbol{a}_m\}\subset\mathcal{T}\subset {\mathbb R}^2$, where $\mathcal{T}$ is defined in \eqref{def-M-set}.
	The set $\mathcal{A}$ is an optimal tight frame
	if and only if $\{f(\boldsymbol{a}_1),\dotsc,f(\boldsymbol{a}_m)\}$ is the {edge vector set} of an optimal $m$-gon $P\in\mathcal{P}_m$, where $f(\cdot )$ is defined in \eqref{def-F-func}.
	In either case, we have
	\begin{equation*}
		\beta_{\boldsymbol{A}} = \frac{1}{\sqrt{1-2\cdot r({P})}},
	\end{equation*}
	where $\boldsymbol{A}=[\boldsymbol{a}_1,\ldots, \boldsymbol{a}_m]^T$ and $r(P)$ is defined in \eqref{eq:rp}.
\end{theorem}


\begin{remark}\label{example-m4}
	Take $m=4$ as an example. Let $P\subset\mathbb{R}^2$ be the convex $4$-gon with the edge set $\{\boldsymbol{e}_1,\boldsymbol{e}_2,\boldsymbol{e}_3,\boldsymbol{e}_4\}$, where
	$$
	\boldsymbol{e}_1=\begin{pmatrix}
		1 \\
		0
	\end{pmatrix},\quad
	\boldsymbol{e}_2=
	\begin{pmatrix}
		\cos \frac{4\pi}{3}  \\
		\sin \frac{4\pi}{3} 
	\end{pmatrix},\quad
	\boldsymbol{e}_3=
	2\sin \frac{\pi}{12}  \begin{pmatrix}
		\cos \frac{3\pi}{4}  \\
		\sin \frac{3\pi}{4} 
	\end{pmatrix},\quad
	\boldsymbol{e}_4=
	2\sin \frac{\pi}{12}  \begin{pmatrix}
		\cos \frac{7\pi}{12}  \\
		\sin \frac{7\pi}{12} 
	\end{pmatrix}.
	$$
	It is known that $P$ is the unique optimal $4$-gon, and we have $r(P)=\frac{1}{2+\sqrt{6}-\sqrt{2}}$ \cite{taylor1953some,tamvakis1987perimeter}.
	By Theorem \ref{main} we see that the matrix $\boldsymbol{A}=[f^{-1}(\boldsymbol{e}_1),\ldots, f^{-1}(\boldsymbol{e}_4)]^T\in\mathbb{R}^{4\times 2}$ has the minimal condition number $\beta_{\boldsymbol{A}}=\sqrt{1+\frac{\sqrt{2}}{2}+\frac{\sqrt{6}}{2}}\approx 1.71$.
	The optimal convex $4$-gon $P$ and the rows of $\boldsymbol{A}$ are both illustrated in Figure \ref{fig:quadrilateral}. 
\end{remark}

\begin{figure}[htbp]
	\centering
	
	\begin{minipage}[b]{0.48\textwidth}
		\centering
		\includegraphics[width=0.62\linewidth]{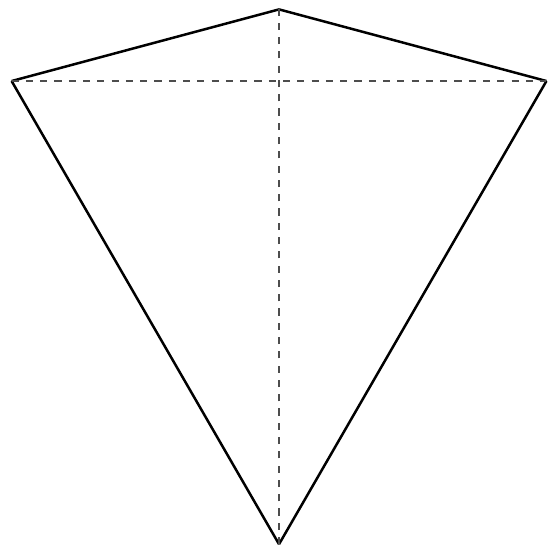}
		\par\smallskip
		{(a)}
	\end{minipage}
	\hfill
	\begin{minipage}[b]{0.48\textwidth}
		\centering
		\includegraphics[width=0.86\linewidth]{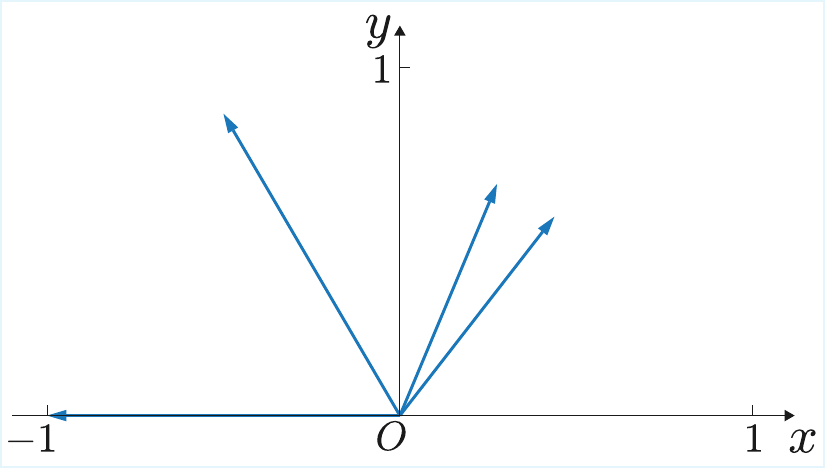}
		\par\smallskip
		{(b)}
	\end{minipage}
	
	\caption{(a): Optimal quadrilateral. The dashed lines connect pairs of vertices at maximal distance; (b): Optimal frame consisting of $4$ vectors in $\mathcal{T}$.}
	\label{fig:quadrilateral}
\end{figure}



With the help of Lemma \ref{optistight2}, Theorem \ref{knownresults} and Theorem \ref{main}, we present a proof of Theorem \ref{mincond}.


\begin{proof}[Proof of Theorem \ref{mincond}]
	By Lemma \ref{optistight2} we have
	\begin{equation*}
		\min _{\boldsymbol{A} \in \mathbb{R}^{m \times 2}} 
		\beta_{\boldsymbol{A}} 
		=\min _{\boldsymbol{A} \in \mathbb{R}^{m \times 2}, \boldsymbol{A}^T\boldsymbol{A}=\boldsymbol{I}_2} 
		\beta_{\boldsymbol{A}} .
	\end{equation*}
	Let $\boldsymbol{B} \in \mathbb{R}^{m \times 2}$ be such that $\boldsymbol{B}^T\boldsymbol{B}=\boldsymbol{I}_2$ and the condition number of $\boldsymbol{B}$ is minimal.
	By Theorem \ref{main}, there is an optimal convex $m$-gon ${P} \subset \mathbb{R}^2$ such that
	\begin{equation}\label{eq:xu10}
		\min _{\boldsymbol{A} \in \mathbb{R}^{m \times 2}} 
		\beta_{\boldsymbol{A}} 
		=\beta_{\boldsymbol{B}}
		=\frac{1}{\sqrt{1-2\cdot r(P)}}.
	\end{equation}
	
	(i) For any integers $m\geq 3$ which has an odd factor, by Theorem \ref{knownresults} (i) we have $r(P) = \frac{1}{2m \cdot \sin \frac{\pi}{2 m}}$.
	Hence, combining with \eqref{eq:xu10}, we obtain 
	\[
	\min _{\boldsymbol{A} \in \mathbb{R}^{m \times 2}} 
	\beta_{\boldsymbol{A}} = \frac{1}{\sqrt{1-\frac{1}{m \cdot \sin \frac{\pi}{2 m}}}}.
	\]
	
	(ii) For $m=2^s$ with $s\geq 2$, by Theorem \ref{knownresults} (ii) we have $r(P) > \frac{1}{2m \cdot \sin \frac{\pi}{2 m}}$. Hence, combining with \eqref{eq:xu10}, we obtain 
	\[
	\min _{\boldsymbol{A} \in \mathbb{R}^{m \times 2}} 
	\beta_{\boldsymbol{A}} > \frac{1}{\sqrt{1-\frac{1}{m \cdot \sin \frac{\pi}{2 m}}}}.
	\]
	
	(iii) Recall that $P_m$ is the regular $m$-gon. Theorem \ref{knownresults} (iii) shows that $r(P_m)=r(P)$ if and only if $m$ is odd, and $P_m$ is uniquely optimal if and only if $m$ is a prime. 
	According to Theorem \ref{main}, the matrix constructed from $P_m$ corresponds to the harmonic frame $\boldsymbol{E}_m$; in particular, $\boldsymbol{E}_m$ consists of the vectors $f^{-1}(\boldsymbol{e}_j)$ obtained from the edge vectors $\boldsymbol{e}_j$ of $P_m, j = 1,\ldots, m$.
	Combining Theorem \ref{main} and Theorem  \ref{knownresults}, we conclude  that $\boldsymbol{E}_m$ has the minimal condition number if and only if $m$ is odd. Moreover,  $\boldsymbol{E}_m$ is the unique matrix with minimal condition number if and only if $m$ is a prime.
\end{proof}

\subsection{Proof of Theorem \ref{main}}
In order to establish Theorem \ref{main}, we first present Lemma \ref{tight_gon} as a preliminary step. 
It has been proven in \cite{han2007frames} that $\{\boldsymbol{a}_1,\ldots, \boldsymbol{a}_m\}\subset {\mathbb R}^2$ forms a tight frame if and only if the diagram vectors $\{f(\boldsymbol{a}_1),\dotsc,f(\boldsymbol{a}_m)\}$ sum to zero, yielding a {\em closed} polygon in $\mathbb{R}^2$. Lemma \ref{tight_gon} further clarifies the relationship between tight frames in $\mathbb{R}^2$ consisting of nonzero vectors with pairwise distinct directions and strictly {\em convex} polygons. Moreover, it connects the condition number of the matrix associated with such a frame to the diameter-to-perimeter ratio of the corresponding polygon.



\begin{lemma}
	\label{tight_gon}
	Assume that 
	$\mathcal{A} = \{\boldsymbol{a}_1,\dotsc,\boldsymbol{a}_{m}\} \subset  \mathcal{T}$, where $\mathcal{T}$ is defined in \eqref{def-M-set}.
	Let $E=\{\boldsymbol{e}_1,\dotsc,\boldsymbol{e}_{m}\}$, where each $\boldsymbol{e}_j =f(\boldsymbol{a}_j) $ with $f$ defined in \eqref{def-F-func}.
	Then the following hold:
	\begin{enumerate}
		\item[\rm (i)] $\mathcal{A} $ is a tight frame consisting of nonzero vectors with pairwise distinct directions if and only if $E$ is the edge set of a strictly convex $m$-gon $ {P} \in \mathcal{P}_m$. 
		\item[\rm (ii)] Assume that $\mathcal{A} $ forms a tight frame. If the vectors in $\mathcal{A}$ are nonzero and have pairwise distinct directions, then 
		\begin{equation*}
			\beta_{\boldsymbol{A}} = \frac{1}{\sqrt{1-2\cdot r({P})}},		
		\end{equation*}
		where $\boldsymbol{A} = [\boldsymbol{a}_1,\dotsc,\boldsymbol{a}_{m}]^T$ and $P$ is the strictly convex $m$-gon defined in {\rm (i)}; otherwise, the condition number of $\boldsymbol{A}$ is not the minimum value attainable among all $m \times 2$ matrices.
	\end{enumerate}
\end{lemma}

\begin{proof}
	For clarity, we write $\boldsymbol{a}_j = t_j(\cos{\phi_j},\sin{\phi_j})^{T}$ for each $j\in[m]$, where $t_j  \geq 0$ and  
	$\phi_j \in (0,\pi]$.
	
	(i) Recall that $\{\boldsymbol{a}_1,\dots,\boldsymbol{a}_{m}\}$ is a tight frame if and only if $\sum_{j=1}^{m}\boldsymbol{e}_{j}=\boldsymbol{0}$ \cite[Lemma 4.1]{han2007frames}. 
	Note that $\boldsymbol{e}_j=f(\boldsymbol{a}_j)=\boldsymbol{0}$ if and only if $\boldsymbol{a}_j=\boldsymbol{0}$.
	Also note that for any $i \neq j$, the vectors $\boldsymbol{e}_i=f(\boldsymbol{a}_i)$ and $\boldsymbol{e}_j=f(\boldsymbol{a}_j)$  have the same direction if and only if $\boldsymbol{a}_i$ and $\boldsymbol{a}_j$ have the same direction. 
	Therefore, $\{\boldsymbol{a}_1,\dots,\boldsymbol{a}_{m}\}$ is a tight frame of nonzero vectors with pairwise distinct directions if and only if $\sum_{j=1}^{m}\boldsymbol{e}_{j}=\boldsymbol{0}$ and $\boldsymbol{e}_1,\dots,\boldsymbol{e}_{m}$ are nonzero vectors with pairwise distinct directions. By Corollary \ref{Minkowski2}, this is further equivalent to $E$ being the edge set of a strictly convex $m$-gon $ {P} \in \mathcal{P}_m$.
	
	(ii) 
	Since $\mathcal{A}$ is a tight frame, by Lemma \ref{eqform} we have 
	\begin{equation}
		\label{beta_equi2}
		\beta_{\boldsymbol{A}} =
		\left(1-\frac{1}{C} \max\limits_{\boldsymbol{x},\boldsymbol{y} \in \mathbb{S}^1 ,\langle \boldsymbol{x},\boldsymbol{y} \rangle=0}
		\sum_{j=1}^m \left|\boldsymbol{x}^T \boldsymbol{a}_j \boldsymbol{a}_j^T \boldsymbol{y}\right| 
		\right)^{-\frac{1}{2}},
	\end{equation}
	where $C = \frac{1}{2}\sum_{j=1}^{m}{\|\boldsymbol{a}_j\|_2^2}=\frac{1}{2}\sum_{j=1}^{m} \|\boldsymbol{e}_{j}\|_2 $.
	Let 
	\[
	\boldsymbol{x}=(\cos\theta,\sin \theta)^{T}, \quad \boldsymbol{y}=(-\sin\theta,\cos \theta)^{T}, \quad \text{and} \quad \boldsymbol{u}=(\sin2\theta,\cos 2\theta)^{T},
	\]
	where $\theta \in [0,2\pi)$.
	A direct calculation shows that for each $j\in[m]$ 
	\begin{equation*}
		\left|\boldsymbol{x}^T \boldsymbol{a}_j \boldsymbol{a}_j^T \boldsymbol{y}\right|
		=\frac{t_j^2}{2}|\sin(2\theta-2\phi_j)|
		=\frac{1}{2}|\langle \boldsymbol{e}_j, \boldsymbol{u} \rangle |.	
	\end{equation*}
	Therefore, we have 
	\begin{equation}\label{eq:xu6}
		\max\limits_{\boldsymbol{x},\boldsymbol{y} \in \mathbb{S}^1 ,\langle \boldsymbol{x},\boldsymbol{y} \rangle=0}
		\sum_{j=1}^m \left|\boldsymbol{x}^T \boldsymbol{a}_j \boldsymbol{a}_j^T \boldsymbol{y}\right|
		=\max\limits_{\theta \in [0,2\pi)}	
		\sum_{j=1}^m \frac{t_j^2}{2}|\sin(2\theta-2\phi_j)|
		=\frac{1}{2}\cdot \max\limits_{\boldsymbol{u} \in \mathbb{S}^1 }	
		\sum_{j=1}^m |\langle \boldsymbol{e}_j, \boldsymbol{u} \rangle |.
	\end{equation}
	Substituting \eqref{eq:xu6} and $C=\frac{1}{2}\sum_{j=1}^{m} \|\boldsymbol{e}_{j}\|_2  $ into \eqref{beta_equi2}, we obtain
	\begin{equation}
		\label{eq:xu7}
		\beta_{\boldsymbol{A}} =
		\left(1-
		\frac{\max\limits_{\boldsymbol{u} \in \mathbb{S}^1 }	
			\sum_{j=1}^m |\langle \boldsymbol{e}_j, \boldsymbol{u} \rangle |}
		{ \sum_{j=1}^{m} \|\boldsymbol{e}_{j}\|_2} 
		\right)^{-\frac{1}{2}}.
	\end{equation}
	If the vectors in $\mathcal{A}$ are nonzero and have pairwise distinct directions, it follows from (i)  that $E$ is the edge set of a strictly convex $m$-gon $ {P} \in \mathcal{P}_m$. 
	Combining \eqref{eq:xu7} with Lemma \ref{tight_gon-form}, we obtain $\beta_{\boldsymbol{A}}=\frac{1}{\sqrt{1-2\cdot r({P})}}$.
	
	We next prove that the condition number of $\boldsymbol{A}$ is not minimal whenever $\mathcal{A}$ contains zero vectors or vectors that have the same direction.
	Assume $\mathcal{A}$  contains such vectors. Then so does the set $E$. 
	By deleting zero vectors from $E$ and replacing all vectors in the same direction with their sum, we obtain a new set $\widehat{E}=\{\widehat{\boldsymbol{e}}_1,\dotsc,\widehat{\boldsymbol{e}}_{k}\}$ with $k < m$.
	The vectors in $\widehat{E}$ are nonzero and have pairwise distinct directions. 
	Note that $\sum_{j=1}^{k}\widehat{\boldsymbol{e}}_j = \sum_{j=1}^{m}\boldsymbol{e}_j = 0$. 
	By Corollary \ref{Minkowski2}, there exists a strictly convex $k$-gon $ \widehat{P}\in\mathcal{P}_k$ with the edge set $\widehat{E}$. 
	Note that $\sum_{j=1}^{m} \|\boldsymbol{e}_{j}\|_2=\sum_{j=1}^{k} \|\widehat{\boldsymbol{e}}_{j}\|_2$ and $\sum_{j=1}^m |\langle \boldsymbol{e}_j, \boldsymbol{u} \rangle |=\sum_{j=1}^k |\langle \widehat{\boldsymbol{e}}_j, \boldsymbol{u} \rangle |$ for any $\boldsymbol{u}\in\mathbb{S}^1$. 
	Combining with \eqref{eq:xu7} and Lemma \ref{tight_gon-form}, we have 
	\begin{equation}\label{eq:xu8}
		\beta_{\boldsymbol{A}} = \frac{1}{\sqrt{1-2\cdot r(\widehat{P})}}.
	\end{equation}
	Recall that any optimal convex m-gon is strictly convex \cite{reinhardt1922extremale}. Since $\widehat{P}\in\mathcal{P}_k$ is a convex $m$-gon but not strictly convex (it has a flat interior angle), it cannot be optimal. Hence, there exists a strictly convex $m$-gon $Q\in\mathcal{P}_m$ such that $r(\widehat{P})> r({Q})$. Denote the edge set of $Q$ by $\{\boldsymbol{q}_1,\ldots, \boldsymbol{q}_m\}$.
	Let $\boldsymbol{B}=[f^{-1}(\boldsymbol{q}_1),\ldots,f^{-1}(\boldsymbol{q}_m)]^T\in\mathbb{R}^{m\times 2}$, where $f(\cdot )$ is defined in \eqref{def-F-func}.
	According to (i), the set $\{f^{-1}(\boldsymbol{q}_1),\ldots,f^{-1}(\boldsymbol{q}_m)\}$ consists of nonzero vectors with pairwise distinct directions. Thus, we have 
	\begin{equation}\label{eq:xu9}
		\beta_{\boldsymbol{B}} = \frac{1}{\sqrt{1-2\cdot r({Q})}}.
	\end{equation}
	Recall that $r(\widehat{P})> r({Q})$. Combining with \eqref{eq:xu8} and \eqref{eq:xu9}, we have $\beta_{\boldsymbol{A}} > \beta_{\boldsymbol{B}}$.
	Therefore, the condition number of $\boldsymbol{A}$ is not minimal. This completes the proof.
\end{proof}


We are now prepared to present the proof of Theorem \ref{main}.

\begin{proof}[Proof of Theorem \ref{main}]
	(i)
	Assume that $\mathcal{A}$ is an optimal tight frame, i.e., the matrix $\boldsymbol{A}=[\boldsymbol{a}_1,\ldots, \boldsymbol{a}_m]^T$ has the minimal condition number and its rows form a tight frame. It follows from Lemma \ref{tight_gon} (ii) that $\mathcal{A}$ consists of nonzero vectors with pairwise distinct directions. Then Lemma \ref{tight_gon} (i) shows that the set $\{f(\boldsymbol{a}_1),\dotsc,f(\boldsymbol{a}_m)\}$ is the edge set of a strictly convex $m$-gon $ {P}\in\mathcal{P}_m$, and Lemma \ref{tight_gon} (ii) further yields that $\beta_{{\boldsymbol{A}}}=\frac{1}{\sqrt{1-2\cdot r(P)}}$. 
	We prove that $P$ is optimal by contradiction. 
	Suppose, to the contrary, that $P$ is not optimal. 
	Then there exists a strictly convex $m$-gon $ {Q}\in\mathcal{P}_m$  such that $r(Q)<r(P)$. 
	Denote the edge vector set of $Q$ by $\{\boldsymbol{q}_1,\ldots, \boldsymbol{q}_m\}$.
	Let $\boldsymbol{B}=[f^{-1}(\boldsymbol{q}_1),\ldots,f^{-1}(\boldsymbol{q}_m)]^T\in\mathbb{R}^{m\times 2}$.
	By Lemma \ref{tight_gon} (i) and (ii), we have $\beta_{\boldsymbol{B}} = \frac{1}{\sqrt{1-2\cdot r({Q})}}$.
	Since $r(Q)<r(P)$, we have $\beta_{\boldsymbol{B}}<\beta_{\boldsymbol{A}}$. This contradicts with the fact that $\boldsymbol{A}$ has the minimal condition number. Hence, $P$ is optimal.

	(ii) Assume that $\{f(\boldsymbol{a}_1),\dotsc,f(\boldsymbol{a}_m)\}$ is the edge set of an optimal $m$-gon $P\in\mathcal{P}_m$. 
	Then $P$ is strictly convex.
	According to Lemma \ref{tight_gon} (i) and (ii), the set $\{\boldsymbol{a}_1,\dotsc,\boldsymbol{a}_{m}\}$ is a tight frame, and we have $\beta_{{\boldsymbol{A}}}=\frac{1}{\sqrt{1-2\cdot r(P)}}$.
	We prove that $\boldsymbol{A}$ has the minimal condition number by contradiction. 
	Assume that $\beta_{\boldsymbol{B}}<\beta_{\boldsymbol{A}}$, where $\boldsymbol{B}\in\mathbb{R}^{m\times 2}$ is a matrix with the minimal condition number and its rows form a tight frame. 
	By Lemma \ref{tight_gon} (ii), the row vectors of $\boldsymbol{B}$ are nonzero and have
	pairwise distinct directions, and there is a strictly convex $m$-gon $ {Q}\in\mathcal{P}_m$ such that $ \beta_{{\boldsymbol{B}}}=\frac{1}{\sqrt{1-2\cdot r(Q)}}$.
	Since $\beta_{\boldsymbol{B}}<\beta_{\boldsymbol{A}}$, we have $r(Q)<r(P)$. 
	This contradicts with $P$ being an optimal convex $m$-gon. 
	Hence, $\boldsymbol{A}$ has the minimal condition number, and its rows form a tight frame.
\end{proof}

\section{Characterization of optimal frames and optimal polygons: Proof of Theorem \ref{struct}}

The aim of this section is to prove Theorem \ref{struct}, which characterizes all optimal frames in $\mathbb{R}^2$ with $m$ vectors when $m \geq 3$ has an odd factor. Our approach continues to draw on tools from the perimeter-maximizing isodiametric problem: we first analyze the structure of all optimal $m$-gons, and the proof of Theorem \ref{struct} then follows by combining this analysis with Theorem \ref{main}.

\subsection{Characterization of optimal polygons}
Let ${\mathcal{E}}_m$ denote the collection of all optimal convex $m$-gons in $\mathbb{R}^2$. The characterization of optimal polygons $P \in {\mathcal{E}}_m$ begins with the introduction of the following problem concerning the discrepancy of roots of unity. This problem has already appeared in \eqref{eq:partition-inThm}, and here we state it formally. 
\begin{problem}\label{pr2}
	Let $m \ge 3$ be an integer. Define $\zeta_m := e^{\ii\pi/m}$.
	Determine all vectors $\boldsymbol{\varepsilon} =(\varepsilon_1,\dotsc,\varepsilon_{m})^T\in\{\pm1\}^m$ such that
	\begin{equation}\label{eq:partition2}
		g(\boldsymbol{\varepsilon}):=\sum_{j=1}^{m} \varepsilon_j \zeta_m^j=0.	
	\end{equation}
	
	%
\end{problem}

In the following, we demonstrate a close connection between the above discrepancy problem and the optimal polygons.
The proof of Theorem \ref{th:polygon} is postponed to the end of this section.

\begin{theorem}\label{th:polygon}
	Assume that $m\geq 3$ be an integer with an odd factor. 
	Set $\boldsymbol{\mu}_j:=(\cos\frac{j\pi}{m}, \sin\frac{j\pi}{m})^T$ for each $j \in [m]$. Then we have
	\begin{enumerate}
		\item[\rm (i)] If $\boldsymbol{\varepsilon}=(\varepsilon_1,\dotsc,\varepsilon_{m})^T\in\{\pm1\}^m$ satisfies \eqref{eq:partition2}, then there exists an optimal polygon $P\in {\mathcal E}_m$ with the edge vectors $\{\varepsilon_1\boldsymbol{\mu}_1,\ldots,\varepsilon_m\boldsymbol{\mu}_m\}$. Such a polygon $P$ is unique up to translation.
		\item[\rm (ii)] Conversely, for any optimal polygon $P\in {\mathcal E}_m$, there exists a vector $\boldsymbol{\varepsilon}=(\varepsilon_1,\dotsc,\varepsilon_{m})^T\in\{\pm1\}^m$ satisfying \eqref{eq:partition2}, such that the edge vectors of $P$ are $\{\varepsilon_1\boldsymbol{\mu}_1,\ldots,\varepsilon_m\boldsymbol{\mu}_m\}$, after a proper scaling and rotation.
	\end{enumerate}
\end{theorem}


\begin{remark}
	For a vector $\boldsymbol{\varepsilon} =(\varepsilon_1,\dotsc,\varepsilon_{m})^T\in\{\pm1\}^m$, we define the shifting operator $\mathcal{S}$ and the reflection operator $\mathcal{R}$ as 
	\begin{equation}    \label{def:shift}
		\mathcal{S}(\boldsymbol{\varepsilon} )=(\varepsilon_{2},\ldots,\varepsilon_{m},-\varepsilon_{1})^T
		\quad\text{and}\quad
		\mathcal{R}(\boldsymbol{\varepsilon} )=
		(\varepsilon_{m},\varepsilon_{m-1},\ldots,\varepsilon_{1})^T.
	\end{equation}
	For each $i\in\{1,2\}$, let $\boldsymbol{\varepsilon}_i\in\{\pm1\}^m$ be a vector satisfying \eqref{eq:partition2}, and let $P_i\in {\mathcal E}_m$ be the corresponding optimal polygon with the edge set as described in Theorem \ref{th:polygon} {\rm (i)}. A simple calculation shows that
	\begin{enumerate}
		\item[{\rm (i)}] if $P_2$ is obtained from $P_1$ by a proper translation and scaling, then $\boldsymbol{\varepsilon}_{2}=\boldsymbol{\varepsilon}_{1}$;
		\item[{\rm (ii)}] 
		if $P_2$ is obtained from $P_1$ via a clockwise rotation of $\frac{k\pi}{m}$ for some integer $k\in [2m]$, then $\boldsymbol{\varepsilon}_{2}=\mathcal{S}^k(\boldsymbol{\varepsilon}_{1} )$;	
		
		\item[{\rm (iii)}]   if $P_2$ is obtained from $P_1$ by a reflection with respect to the $y$-axis, then $\boldsymbol{\varepsilon}_{2}=\mathcal{S}(\mathcal{R}(\boldsymbol{\varepsilon}_{1} ))$.
	\end{enumerate}
	Since we treat polygons as equivalent under scaling, translations, rotations, reflections, or any combination of these transformations, we define two vectors $\boldsymbol{\varepsilon}_1,\boldsymbol{\varepsilon}_2\in \{\pm1\}^m$ to be equivalent if one can be obtained from the other by applying a finite sequence of the shifting operator $\mathcal{S}$ and the reflection operator $\mathcal{R}$. A simple calculation shows that ${\mathcal R}^2={\mathcal S}^{2m}={\rm Id}$, where ${\rm Id}$ denotes the identity operator. Furthermore, we have ${\mathcal S}^k{\mathcal R}={\mathcal R}{\mathcal S}^{2m-k}$ for any integer $k\in [1, 2m]$. Hence, $\boldsymbol{\varepsilon}_1$ is equivalent to $\boldsymbol{\varepsilon}_2$ if and only if $ {\mathcal R}^{k_1}{\mathcal S}^{k_2} (\boldsymbol{\varepsilon}_1)=\boldsymbol{\varepsilon}_2$ for some $ k_1 \in \{0,1\}$ and $k_2\in \{0,1,\ldots,2m-1\}$. Under this convention, the number of distinct solutions to Problem \ref{pr1} equals that of Problem \ref{pr2}.
\end{remark}

\begin{remark}
	\label{algorithm-poly}
	Let $m \geq 3$ be an integer with an odd factor. Recall the brute-force algorithm mentioned in Remark \ref{algorithm}. In fact, by Theorem \ref{th:polygon}, it can also be used to generate all optimal polygons in ${\mathcal E}_m$, up to scaling, translation, rotation, and reflection.
	We again take $m=12$ as an example. 
	After checking all vectors $\boldsymbol{\varepsilon}\in\{\pm1\}^{12}$ for which \eqref{eq:partition2} holds, by Theorem \ref{th:polygon}, we identify all optimal polygons in ${\mathcal E}_{12}$ (up to scaling, translation, rotation and reflection), shown in Figure \ref{fig:dodecagons} {\rmfamily\slshape (a1)} and {\rmfamily\slshape (a2)}. 
	Note that this algorithm for constructing all optimal polygons in ${\mathcal E}_m$ differs from the one proposed in \cite{mossinghoff2011enumerating}.
	A detailed comparison of the efficiency of the two algorithms is left for future work.
\end{remark}

\begin{figure}[htbp]
	\centering
	
	\begin{minipage}[b]{0.48\textwidth}
		\centering
		\includegraphics[width=0.62\linewidth]{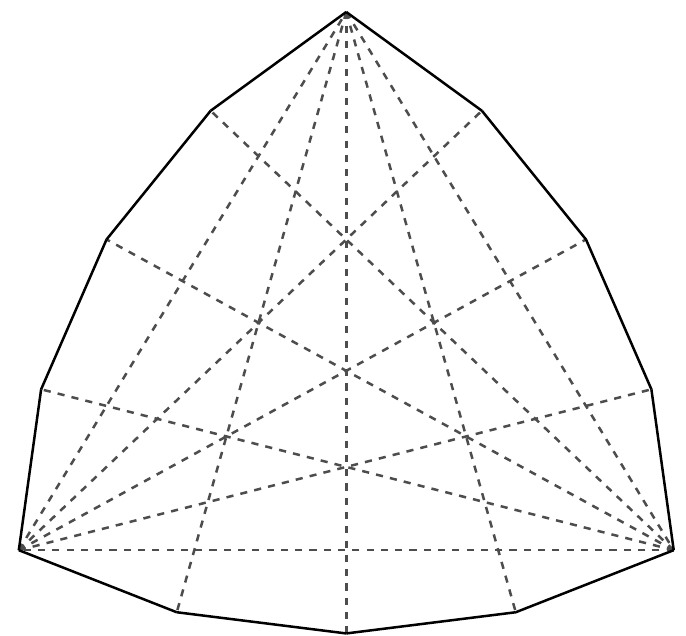}
		\par\smallskip
		{(a1)}
	\end{minipage}
	\hfill
	\begin{minipage}[b]{0.48\textwidth}
		\centering
		\includegraphics[width=0.62\linewidth]{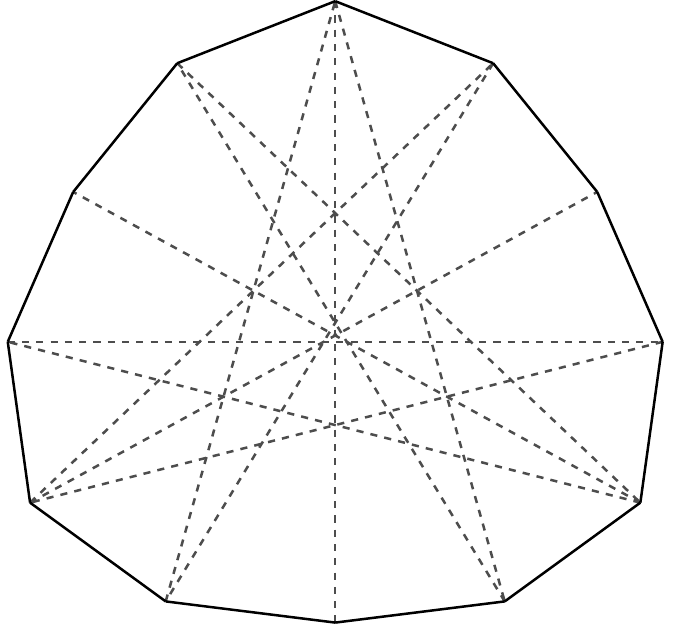}
		\par\smallskip
		{(a2)}
	\end{minipage}
	
	\par\bigskip
	
	\begin{minipage}[b]{0.48\textwidth}
		\centering
		\includegraphics[width=0.70\linewidth]{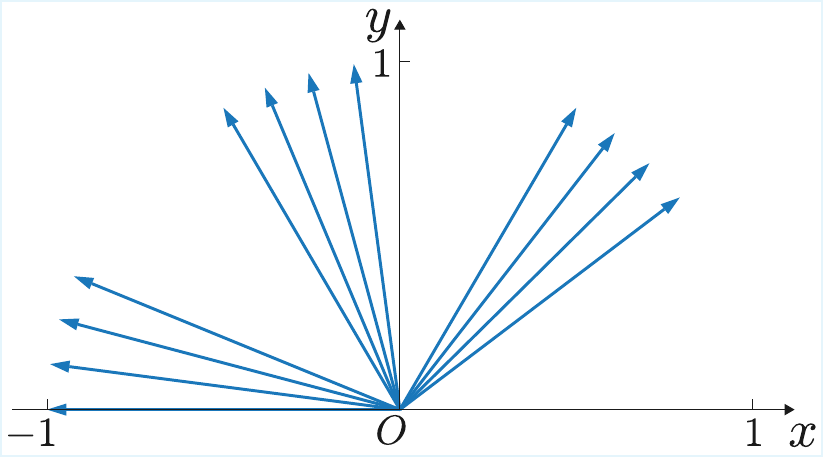}
		\par\smallskip
		{(b1)}
	\end{minipage}
	\hfill
	\begin{minipage}[b]{0.48\textwidth}
		\centering
		\includegraphics[width=0.70\linewidth]{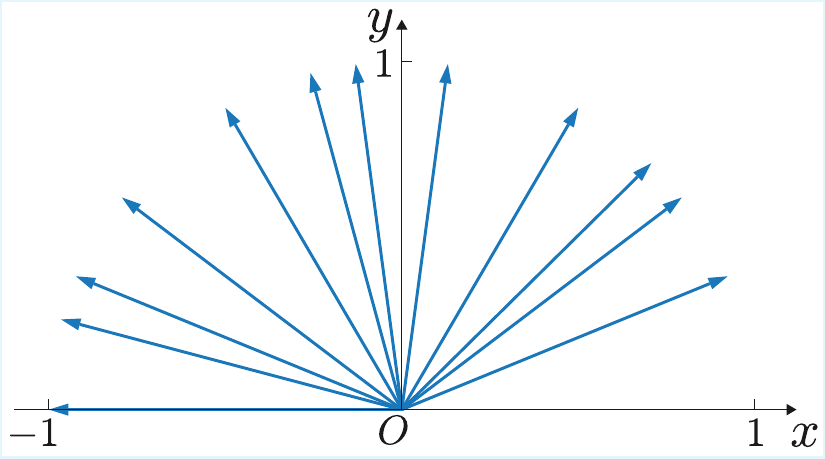}
		\par\smallskip
		{(b2)}
	\end{minipage}
	
	\caption{(a1), (a2): All optimal dodecagons. The dashed lines connect pairs of vertices at maximal distance; (b1), (b2): All optimal frames consisting of $12$ vectors in $\mathcal{T}$, corresponding to (a1) and (a2) respectively.}
	\label{fig:dodecagons}
\end{figure}

\subsection{Proof of Theorem \ref{struct}}
With the help of Theorem \ref{th:polygon}, we give a proof of Theorem \ref{struct}. 
%


\begin{proof}[Proof of Theorem \ref{struct}]
	
	(i) 
	We first prove that $\mathcal{A}$ forms a tight frame by contradiction.
	Suppose, for contradiction, that $\mathcal{A}$ is not a tight frame.
	Set $\boldsymbol{A} = [\boldsymbol{a}_1,\ldots, \boldsymbol{a}_m]^T$ with $\boldsymbol{A} \in \mathbb{R}^{m \times 2}$.
	By Theorem \ref{LAH}, we have
	$$
	(U_{\boldsymbol{A}})^2 = \|\boldsymbol{A}\|_2^2 =  \|\boldsymbol{A}^T\boldsymbol{A}\|_2 > \frac{1}{2}\cdot \operatorname{Tr}(\boldsymbol{A}^T\boldsymbol{A}) = \frac{1}{2}\cdot \operatorname{Tr}(\boldsymbol{A}\boldsymbol{A}^T) = \frac{1}{2}\sum_{j=1}^{m}\|\boldsymbol{a}_j\|_2^2.
	$$
	Combining with Lemma \ref{lowerBound}, we derive that
	\begin{equation}
		\label{sueq}
		\beta_{\boldsymbol{A}} = \frac{U_{\boldsymbol{A}}}{L_{\boldsymbol{A}}} > \frac{1}{\sqrt{1-\frac{1}{m \cdot \sin \frac{\pi}{2 m}}}}.
	\end{equation}
	Corollary \ref{mincond} {\rm (i)} shows that 
	$\min _{\boldsymbol{M} \in \mathbb{R}^{m \times 2}} 
	\beta_{\boldsymbol{M}} = \frac{1}{\sqrt{1-\frac{1}{m \cdot \sin \frac{\pi}{2 m}}}}$. Combining with \eqref{sueq}, we obtain that $\beta_{\boldsymbol{A}} > \min _{\boldsymbol{M} \in \mathbb{R}^{m \times 2}} 
	\beta_{\boldsymbol{M}}$, meaning that $\mathcal{A}$ is not optimal. This contradicts with $\mathcal{A}\in\mathcal{T}_m$. Therefore, $\mathcal{A}$ forms a tight frame.
	
	We next prove that all vectors in ${\mathcal{A}}$ have the same norm.
	By Theorem \ref{main}, there is an optimal strictly convex $m$-gon $P$ with the edge set $\{f(\boldsymbol{a}_1),\dotsc,f(\boldsymbol{a}_m)\}$.
	Since $m \geq 3$ is an integer with an odd factor, Theorem \ref{knownresults} (i) implies that $P$ is equilateral, i.e., $\|f(\boldsymbol{a}_1)\|_2=\|f(\boldsymbol{a}_2)\|_2=\cdots=\|f(\boldsymbol{a}_m)\|_2$. Consequently, according to the definition of $f(\cdot )$ in \eqref{def-F-func}, all vectors in $\mathcal{A}$ have equal norm.

	(ii) 
	If $\boldsymbol{\varepsilon}=(\varepsilon_1,\ldots,\varepsilon_m)^T\in\{\pm1\}^m$ satisfies \eqref{eq:partition-inThm}, by Theorem \ref{th:polygon}, there exists an optimal polygon $P\in {\mathcal E}_m$ with the edge vectors $\{\varepsilon_j(\cos(\frac{j}{m}\pi), \sin(\frac{j}{m}\pi))^T: j\in [m]\} $.
	Note that for $j \in [m]$ we have
	\begin{equation}
		\begin{aligned}\label{nvj}
			f(\nv_j)=(\cos\frac{j\pi}{m}, \sin\frac{j\pi}{m})^T 
			\quad\text{and}\quad
			f(\nv_j^\perp)
			=-(\cos\frac{j\pi}{m}, \sin\frac{j\pi}{m})^T.
		\end{aligned}
	\end{equation}
	Hence, the edge set of $P$ can be rewritten as $\{f(\nv_j):j\in I\} \cup \{f(\nv_j^{\perp}) : j\in [m]\setminus I\}$ where $I=\{j\in[m]:\varepsilon_j=1\}$, and such a polygon $P$ is unique up to translation. 
	Then it follows from Theorem \ref{main} that the set $\{\nv_j:j\in I\} \cup \{\nv_j^{\perp} : j\in [m]\setminus I\}$ is an optimal tight frame.

	(iii) 
	Let ${\mathcal{A}}= \{\boldsymbol{a}_1,\dotsc,\boldsymbol{a}_{m}\} \in  \mathcal{T}_m$. As established in the proof of (i), $\mathcal{A}$ is an optimal tight frame. Hence, by Theorem \ref{main}, there is an optimal strictly convex $m$-gon $P$ with the edge set $\{f(\boldsymbol{a}_1),\dotsc,f(\boldsymbol{a}_m)\}$. 
	According to Theorem \ref{th:polygon}, there exists a vector $\boldsymbol{\varepsilon}=(\varepsilon_1,\ldots,\varepsilon_m)^T\in\{\pm1\}^m$ satisfying \eqref{eq:partition-inThm}, such that after a proper scaling and rotation, the edge set $\{f(\boldsymbol{a}_1),\dotsc,f(\boldsymbol{a}_m)\}$ equals $\{(\cos\frac{j\pi}{m}, \sin\frac{j\pi}{m})^T: j\in I\} \cup \{-(\cos\frac{j\pi}{m}, \sin\frac{j\pi}{m})^T: j\in [m]\setminus I\}$, where $I=\{j\in[m]:\varepsilon_j=1\}$.
	By \eqref{nvj}, this edge set can be rewritten as $\{f(\nv_j):j\in I\} \cup \{f(\nv_j^{\perp}) : j\in [m]\setminus I\}$. 
	Note that the function $f(\cdot )$ is invertible on $\mathcal{T}$. 
	It follows that ${\mathcal A}= \{\boldsymbol{a}_1,\dotsc,\boldsymbol{a}_{m}\}=\{\nv_j:j\in I\} \cup \{\nv_j^{\perp} : j\in [m]\setminus I\}$, up to scaling and rotation. This completes the proof.
\end{proof}

\subsection{Proof of Theorem \ref{th:polygon}}
\label{optgoncharac}
In this subsection we prove Theorem \ref{th:polygon}. We begin by recalling the definition of Reuleaux polygons, which will be used in the sequel.
\begin{definition}
	\label{Reuleaux}
	A Reuleaux polygon is a constant-width curve constructed from circular arcs, all having the same radius.
\end{definition}

Here, we present several key properties of Reuleaux polygons given by \cite{mossinghoff2011enumerating,busemann1958hg, mossinghoff20061,hare2013sporadic,hare2019most}.

\begin{theorem}{\rm\cite{mossinghoff2011enumerating,hare2013sporadic,hare2019most}}
	\label{Reuleauxpro}
	Let $R$ be a Reuleaux polygon. The following holds:
	\begin{enumerate}
		\item[\rm (i)] The boundary of $R$ is composed of $r$ circular arcs, where $r\geq3$ is an odd integer.
		\item[\rm (ii)] Connecting all pairs of vertices at maximal distance from one another in $R$ generates a star polygon $S$, which is a closed planar curve composed of $r$ line segments, each intersecting all the others. The total sum of the interior angles at the vertices of $S$ equals $\pi$.
	\end{enumerate} 
\end{theorem}

\begin{theorem}{\rm\cite{reinhardt1922extremale}}
	\label{opt_gon_eq}
	Suppose $m\geq3$ is an integer with an odd factor. A convex $m$-gon $P$ is optimal if and only if $P$ is equilateral and $P$ can be inscribed in a Reuleaux polygon $R$ with $\mathrm{diam}(R)=\mathrm{diam}(P)$  such that every vertex of $R$ is also a vertex of $P$.
\end{theorem}

\begin{remark}
	\label{giveAnopt}
	When $m \geq 3$ has an odd factor, one can easily construct a convex $m$-gon that satisfies the criteria of Theorem \ref{opt_gon_eq}, using the procedure described in \cite{reinhardt1922extremale, mossinghoff2011enumerating, bingane2021maximal}.
	
\end{remark}

\begin{remark}\label{proofofh}
	Let $P$ be an optimal convex $m$-gon, where $m\geq3$ is an integer with an odd factor.
	By Theorem \ref{opt_gon_eq} and Theorem \ref{Reuleauxpro} (i), 
	$P$ is inscribed in a Reuleaux polygon $R$ with $r$ boundary arcs, where $r$ is an odd integer with $3\leq r\leq m$.
	Label the vertices of $R$ as
	$\boldsymbol{w}_1,\dotsc,\boldsymbol{w}_{r}$ in counterclockwise order, and extend the indexing cyclically by setting $\boldsymbol{w}_{r+i}=\boldsymbol{w}_{i}$ for each $i\in[r]$.
	Define $h=\frac{r-1}{2}$.
	Then for each $i\in[r]$, the vertex $\boldsymbol{w}_i$ is the the center of the  arc $\arc{\boldsymbol{w}_{{i+h}}\boldsymbol{w}_{{i+h+1}}}$ on the boundary of $R$. 
	This can be seen as follows. 
	By Theorem \ref{Reuleauxpro} (ii), 
	we can start at $\boldsymbol{s}_1:=\boldsymbol{w}_{1}$ and construct a star polygon $S$ by sequentially
	connecting pairs of vertices in $R$ at maximal distance in the counterclockwise direction.
	Let $\boldsymbol{s}_i$ denote the $i$-th vertex of $R$ visited during this construction.
	According to \cite[Page 256]{reinhardt1922extremale}, the vertices $\boldsymbol{s}_1, \boldsymbol{s}_{3},\dotsc,\boldsymbol{s}_{r},\boldsymbol{s}_2,\boldsymbol{s}_{4},\dotsc,\boldsymbol{s}_{r-1}$ are ordered counterclockwise.
	Comparing this with the labeling $\boldsymbol{w}_{1},\ldots,\boldsymbol{w}_{r}$, 
	we have
	\begin{equation*}
		\boldsymbol{w}_{1}=\boldsymbol{s}_{1},
		\boldsymbol{w}_{2}=\boldsymbol{s}_{3},
		\ldots, 
		\boldsymbol{w}_{\frac{r+1}{2}}=\boldsymbol{s}_{r},
		\boldsymbol{w}_{\frac{r+1}{2}+1}=\boldsymbol{s}_{2}, \ldots,
		\boldsymbol{w}_{r}=\boldsymbol{s}_{r-1}.	
	\end{equation*}
	Set $\boldsymbol{s}_{0}:=\boldsymbol{s}_{r}$ and $\boldsymbol{s}_{r+1}:=\boldsymbol{s}_{1}$. Note that $\boldsymbol{s}_{i}$ is the center of the arc $\arc{\boldsymbol{s}_{i-1}\boldsymbol{s}_{i+1}}$ for each $i\in[r]$. 
	Translating this to the $\boldsymbol{w}_i$  notation, we conclude that for each $i\in[r]$, $\boldsymbol{w}_{{i}}$ is the center of the arc $\arc{\boldsymbol{w}_{{i+h}}\boldsymbol{w}_{{i+h+1}}}$, where $h=\frac{r-1}{2}$.
\end{remark}

\begin{figure}[htbp]
	\centering
	\includegraphics[width=0.48\linewidth]{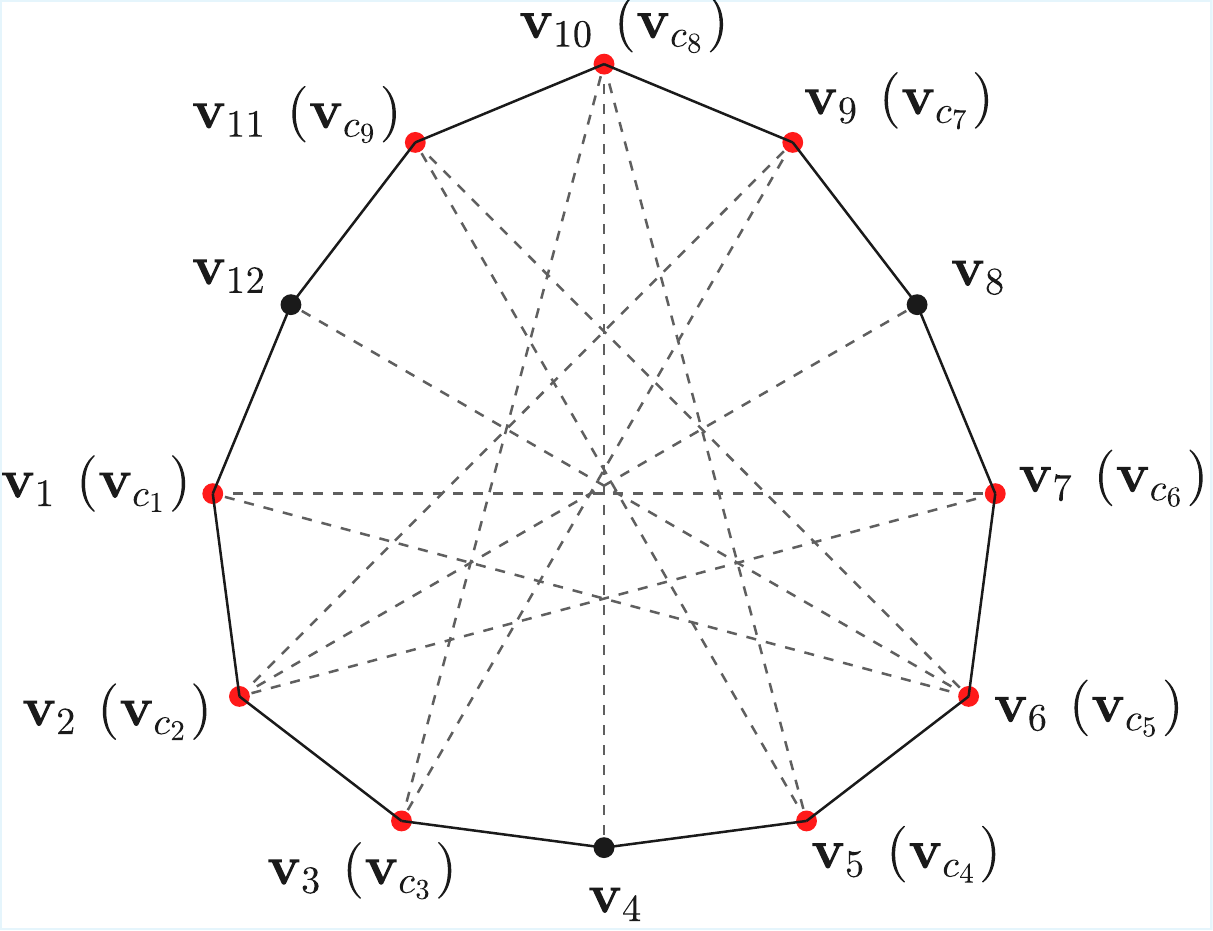}
	\caption{One of the optimal dodecagons as in  Figure \ref{fig:dodecagons} (a2), with vertices $\boldsymbol{v}_{1},\dotsc,\boldsymbol{v}_{12}$. The corresponding Reuleaux polygon consists of $9$ circular arcs with vertices $\boldsymbol{v}_{c_1},\dotsc,\boldsymbol{v}_{c_9}$ indicated in red dots. The dashed lines connect pairs of vertices at maximal distance.}
	\label{fig:dodecagon_illustration}
\end{figure}

The following lemma, motivated by \cite{mossinghoff2011enumerating}, establishes a characterization of the directional angles of the edge vectors for optimal convex polygons.

\begin{lemma}
	\label{angles}
	Assume that $m\geq 3$ is an integer with an odd factor. Let $P$ be a convex $m$-gon with edge vectors $\{\boldsymbol{e}_1,\dotsc\boldsymbol{e}_m\}$. For each $i \in [m]$, define $\alpha_i \in [0, 2\pi)$ as the counterclockwise angle from $\boldsymbol{e}_1$ to $\boldsymbol{e}_i$, and let $\psi_i \in (0, \pi]$ satisfy $\psi_i \equiv \alpha_i \pmod{\pi}$. Let ${\mathcal{E}}_m$ denote the collection of all optimal strictly convex $m$-gons in $\mathbb{R}^2$. 
	If $P \in {\mathcal{E}}_m$ 
	then
	\begin{equation}
		\label{PSI}
		\{\psi_1, \psi_2, \dotsc,\psi_m\} =
		\left\{\frac{j\pi}{m}\; : \; j\in[m] \right\}.
	\end{equation}

\end{lemma}

\begin{proof}
	Assume that $P \in {\mathcal{E}}_m$. It follows from Theorem \ref{opt_gon_eq} that $P$ is equilateral. We now prove that $P$ satisfies \eqref{PSI}. Assume without loss of generality that $\boldsymbol{e}_1,\dotsc,\boldsymbol{e}_{m}$ are arranged counterclockwise. Let $\boldsymbol{v}_1,\dotsc,\boldsymbol{v}_{m}$ be the vertices of $ P$ such that $\boldsymbol{e}_j = \boldsymbol{v}_{j+1} - \boldsymbol{v}_{j}$ for each $j \in [m-1]$ and $\boldsymbol{e}_m = \boldsymbol{v}_{1} - \boldsymbol{v}_{m}$. 
	For each $i\in[m]$, we denote $\boldsymbol{v}_{m+i}=\boldsymbol{v}_{i}$, and let $\theta_i$ be the counterclockwise angle from $\boldsymbol{e}_{i-1}$ to $\boldsymbol{e}_{i}$, where $\boldsymbol{e}_{0}:=\boldsymbol{e}_{m}$.
	Then we have $\alpha_1=0$ and $\alpha_{i}=\alpha_{i-1}+ \theta_{i}$ for each $i\in\{2,\ldots,m\}$.
	Moreover, since $P$ is strictly convex, we have $\theta_i\in(0,\pi)$, for $ i\in[m]$.

	We first calculate each $\theta_i$ by plane geometry.
	By Theorem \ref{opt_gon_eq}, $P$ is inscribed in a Reuleaux polygon $R$ with the same diameter, and every vertex of $R$ is also a vertex of $P$. 
	By Theorem \ref{Reuleauxpro} (i), $R$ consists of $r$ circular arcs, where $3 \leq r \leq m$ is an odd number.
	Denote these circular arcs by $\arc{\boldsymbol{v}_{c_1}\boldsymbol{v}_{c_2}},\arc{\boldsymbol{v}_{c_2}\boldsymbol{v}_{c_3}},\ldots,\arc{\boldsymbol{v}_{c_r}\boldsymbol{v}_{c_1}}$,
	where $1\leq c_1< c_2<\cdots<c_{r-1}<c_r\leq m$.
	Without loss of generality, we may assume that $c_1=1$, because the choice of the edge $\boldsymbol{e}_1$ can be arbitrary and does not affect the outcome of the lemma.
	For convenience, we denote 
	\begin{equation*}
		c_{r+i}:=c_{i}+m,\quad \text{ for } i\in[r].	
	\end{equation*}
	By Remark \ref{proofofh}, for each $i\in[r]$, $\boldsymbol{v}_{c_{i}}$ is the center of the arc $\arc{\boldsymbol{v}_{c_{i+h}}\boldsymbol{v}_{c_{i+h+1}}}$, where $h=\frac{r-1}{2}$ is a fixed integer. 
	See Figure \ref{fig:dodecagon_illustration} for a clear and intuitive illustration. 
	Since $P$ is equilateral and the sum of the interior angles at the vertices of $R$ equals $\pi$, the interior angle of the arc $\arc{\boldsymbol{v}_{c_{i+h}}\boldsymbol{v}_{c_{i+h+1}}}$ centered on $\boldsymbol{v}_{c_i}$ is $(c_{i+h+1}-c_{i+h})\cdot \frac{\pi}{m}$ for each $i\in[r]$.
	By plane geometry, if $i\notin \{c_1,\ldots,c_r\}$ then we have 
	\begin{equation}\label{meq1}
		\theta_i=\frac{\pi}{m}.
	\end{equation}
	Moreover, for any integer $i\in[r]$, we have
	\begin{equation}\label{meq2}
		\theta_{c_i}=\pi - (\pi-\frac{\pi}{m})+(c_{i+h+1}-c_{i+h})\cdot\frac{ \pi}{m}=(c_{i+h+1}-c_{i+h}+1)\cdot \frac{\pi}{m}.	
	\end{equation} 

	We next calculate $\alpha_1,\ldots,\alpha_m$. 
	Recall that  $\alpha_{i}=\alpha_{i-1}+ \theta_{i}$ for each $i\in\{2,\ldots,m\}$.
	Hence, using \eqref{meq1} we have
	\begin{equation}\label{meq5}
		\alpha_{L+c_i}
		=\alpha_{c_i}+ L\cdot \frac{\pi}{m},
		\quad \forall \; i\in[r],
		\ \forall \; 0\leq L\leq c_{i+1}-c_{i}-1.
	\end{equation}
	For any integer $2\leq i\leq r$, using \eqref{meq2} and \eqref{meq5} we have
	\begin{equation}\label{meq4}
		\begin{aligned}
			\alpha_{c_i}
			&=\alpha_{c_i-1}+\theta_{c_i}
			=\alpha_{c_{i-1}}+(c_{i}-c_{i-1}-1)\cdot \frac{\pi}{m}+\theta_{c_i}\\
			&=\alpha_{c_{i-1}}+(c_{i}-c_{i-1}+c_{i+h+1}-c_{i+h})\cdot \frac{\pi}{m}.
		\end{aligned}
	\end{equation}
	Since ${c_1}=1$ and $\alpha_1=0$, repeatedly using \eqref{meq4}, we obtain
	\begin{equation}\label{meq6}
		\begin{aligned}
			\alpha_{c_i}
			=\alpha_{c_{1}}+(c_{i}-c_{1}+c_{i+h+1}-c_{h+2})\cdot \frac{\pi}{m}
			=(c_{i}+c_{i+h+1}-c_{h+2}-1)\cdot \frac{\pi}{m}.
		\end{aligned}
	\end{equation}
	Then we can use \eqref{meq6} and \eqref{meq5} to obtain  $\alpha_l$, $l=1,\ldots,m$.
	
	Now we turn to prove \eqref{PSI}.
	Note that  
	\begin{equation*}
		\begin{aligned}
			\alpha_{c_{h+2}}
			&\overset{(a)}=(c_{h+2}+c_{2h+3}-c_{h+2}-1)\cdot \frac{\pi}{m}
			=\pi+(c_{2}-1)\cdot \frac{\pi}{m}>\pi,\\
			\alpha_{c_{h+2}-1}
			&\overset{(b)}=\alpha_{c_{h+2}}-(c_{2h+3}-c_{2h+2}+1)\cdot \frac{\pi}{m}=\alpha_{c_{h+2}}-c_{2}\cdot \frac{\pi}{m}=\pi- \frac{\pi}{m},
		\end{aligned}
	\end{equation*}
	where $(a)$ follows from \eqref{meq6} and $(b)$  follows from \eqref{meq4}.
	Hence, by the definition of each $\alpha_l$, we have
	\begin{equation}\label{meq9}
		0=\alpha_1<\alpha_2<\cdots<\alpha_{c_{h+2}-1}<\pi <\alpha_{c_{h+2}}<\cdots<\alpha_{m}<2\pi.
	\end{equation}
	Since each $\alpha_l$ is a multiple of $\frac{\pi}{m}$, and by the definition $\psi_i \equiv \alpha_i \pmod{\pi}$, to prove \eqref{PSI}, it is sufficient to show that there are no integers $i,j$ such that $1\leq i<c_{h+2}\leq j\leq m$ and $\alpha_{j}=\alpha_{i}+\pi$. 
	We prove by contradiction.
	Assume that  such integers $i$ and $j$ exist. 
	Then there is an integer $t$ with $h+2\leq t\leq r$ such that $c_{t}\leq j	<c_{t+1}$.
	Using \eqref{meq5}, \eqref{meq6} and $c_{t+h+1}=c_{t-h+r}=c_{t-h}+m$, we obtain
	\begin{equation}\label{meq8}
		\alpha_j=\alpha_{c_t}+ (j-c_t)\cdot \frac{\pi}{m}
		=(c_{t+h+1}+j-c_{h+2}-1)\cdot \frac{\pi}{m}
		=\pi+(c_{t-h}+j-c_{h+2}-1)\cdot \frac{\pi}{m}.
	\end{equation}
	However, note that  
	\begin{equation*}
		\begin{aligned}
			\alpha_{c_{t-h}}&\overset{(a)}=(c_{t-h}+c_{t+1}-c_{h+2}-1)\cdot \frac{\pi}{m}\\
			\alpha_{c_{t-h}-1}
			&\overset{(b)}=\alpha_{c_{t-h}}-(c_{t+1}-c_{t}+1)\cdot \frac{\pi}{m}
			=(c_{t-h}+c_{t}-c_{h+2}-2)\cdot \frac{\pi}{m},
		\end{aligned}
	\end{equation*}
	where $(a)$ follows from \eqref{meq6} and $(b)$  follows from \eqref{meq4}.
	Combining with \eqref{meq8} and $c_{t}\leq j	<c_{t+1}$, we have
	\begin{equation*}
		\alpha_{c_{t-h}-1}< \alpha_i=	\alpha_j-\pi < \alpha_{c_{t-h}}.	
	\end{equation*}
	By \eqref{meq9} we have $c_{t-h}-1<i<c_{t-h}$.
	Since no integer lies between two consecutive integers, 
	such integer $i$ does not exist, and we have a contradiction.
	This completes the proof of \eqref{PSI}. 
\end{proof}

Now we can present a proof of Theorem \ref{th:polygon}.

\begin{proof}[Proof of Theorem \ref{th:polygon}]
	(i) 
	%
	We first prove that there exists a convex $m$-gon $P$ with the edge set $\{\varepsilon_1\boldsymbol{\mu}_1,\dotsc,\varepsilon_m \boldsymbol{\mu}_m\}$. 
	Since $\boldsymbol{\varepsilon}=(\varepsilon_1,\ldots,\varepsilon_m)^T$ satisfies \eqref{eq:partition2}, we have
	\[
	0
	=\sum_{j=1}^{m} \varepsilon_j \zeta_m^j 
	=\sum_{j=1}^m\varepsilon_j\cos\frac{j\pi}{m} 
	+\ii \sum_{j=1}^m\varepsilon_j\sin\frac{j\pi}{m}, 
	\]
	where $\zeta_m = e^{\ii\pi/m}$. It follows that 
	\begin{equation}\label{eq:sum0condition}
		\sum_{j=1}^{m}\varepsilon_j\boldsymbol{\mu}_{j}
		=   (\sum_{j=1}^m\varepsilon_j\cos\frac{j\pi}{m}, \sum_{j=1}^m\varepsilon_j\sin\frac{j\pi}{m})^T
		= \boldsymbol{0}.
	\end{equation}
	Combining with the fact that $\varepsilon_1\boldsymbol{\mu}_1,\ldots,\varepsilon_m\boldsymbol{\mu}_m$ have pairwise distinct directions, we can utilize Corollary \ref{Minkowski2} to obtain a strictly convex $m$-gon $P$ with the edge set $\{\varepsilon_1\boldsymbol{\mu}_1,\dotsc,\varepsilon_m\boldsymbol{\mu}_m\}$, and $P$ is unique up to translation.
	
	We next prove that $P\in {\mathcal E}_m$. 
	For each $i \in [m]$, define $\alpha_i \in [0, 2\pi)$ as the counterclockwise angle from $\varepsilon_1\boldsymbol{\mu}_1$ to $\varepsilon_i\boldsymbol{\mu}_i$, and let $\psi_i \in (0, \pi]$ satisfy $\psi_i \equiv \alpha_i \pmod{\pi}$. 
	Since each $\boldsymbol{\mu}_j=(\cos\frac{j\pi}{m}, \sin\frac{j\pi}{m})^T$, we have 
	\begin{equation}\label{use_PSI}
		\{\psi_1, \psi_2, \dotsc,\psi_m\} =
		\left\{\frac{j\pi}{m}\; : \; j\in[m]\right\}.
	\end{equation}
	Let $Q \in {\mathcal{E}}_m$ be an optimal $m$-polygon with the edge set $\{\boldsymbol{q}_1,\dotsc,\boldsymbol{q}_m\}$ such that $\boldsymbol{q}_1=\varepsilon_1\boldsymbol{\mu}_1$.
	By Theorem \ref{opt_gon_eq} we see that $Q$ is equilateral, so we have $\|\boldsymbol{q}_j\|_2=1$ for each $j\in[m]$.
	%
	For each $i \in [m]$, define $\beta_i \in [0, 2\pi)$ as the counterclockwise angle from $\boldsymbol{q}_1$ to $\boldsymbol{q}_i$, and let $\phi_i \in (0, \pi]$ satisfy $\phi_i \equiv \beta_i \pmod{\pi}$. By Lemma \ref{angles} we have 
	\begin{equation*}
		\{\phi_1, \phi_2, \dotsc,\phi_m\} =  \left\{\frac{j\pi}{m}\; : \; j\in[m] \right\}.
	\end{equation*}
	Therefore, there exists a permutation $\{s_1,s_2,\ldots,s_m\}$ of $[m]$ such that
	$\psi_i=\phi_{s_i}$ for each $i\in[m]$.
	This means that for each $i\in[m]$ we have
	$\alpha_i \equiv \beta_{s_i} \pmod{\pi} $.
	Recall that $\boldsymbol{q}_1 =  \varepsilon_1\boldsymbol{\mu}_1=\varepsilon_1(\cos\frac{\pi}{m}, \sin\frac{\pi}{m})^T$.
	Let $I:=\{i\in[m]\mid  \alpha_i = \beta_{s_i}\}$.
	If $i\in I$, then 
	\begin{equation*}
		\boldsymbol{q}_{s_i}
		=\varepsilon_1(\cos(\frac{\pi}{m}+\beta_{s_i}), \sin(\frac{\pi}{m}+\beta_{s_i}))^T
		=\varepsilon_1(\cos(\frac{\pi}{m}+\alpha_{i}), \sin(\frac{\pi}{m}+\alpha_{i}))^T
		=\varepsilon_i  \boldsymbol{\mu}_i.	
	\end{equation*}
	Otherwise, if $i\in[m]\setminus I$, then $\alpha_i = \beta_{s_i}- \pi$ or $\alpha_i = \beta_{s_i}+ \pi$. So we have
	\begin{equation*}
		\boldsymbol{q}_{s_i}
		=\varepsilon_1(\cos(\frac{\pi}{m}+\beta_{s_i}), \sin(\frac{\pi}{m}+\beta_{s_i}))^T
		=-\varepsilon_1(\cos(\frac{\pi}{m}+\alpha_{i}), \sin(\frac{\pi}{m}+\alpha_{i}))^T
		=-\varepsilon_i  \boldsymbol{\mu}_i.	
	\end{equation*}
	Hence,  
	we have $\{\boldsymbol{q}_1,\dotsc,\boldsymbol{q}_m\}=\{\varepsilon_i\boldsymbol{\mu}_i\}_{i \in I} \cup \{-\varepsilon_i\boldsymbol{\mu}_i\}_{i \in [m]\setminus I}$.
	%
	According to \eqref{R} in Lemma \ref{tight_gon-form}, we have 
	\begin{equation*}
		r( {P}) 
		= \frac{\max_{\boldsymbol{u} \in \mathbb{S}^1}\,\sum_{i=1}^{m} {|\langle \boldsymbol{\mu}_{i},\boldsymbol{u} \rangle|}}{2\sum_{i=1}^{m}\|\boldsymbol{\mu}_i\|_2} 
		= \frac{\max_{\boldsymbol{u} \in \mathbb{S}^1}\,\sum_{i=1}^{m} {|\langle \boldsymbol{q}_{i},\boldsymbol{u} \rangle|}}{2\sum_{i=1}^{m}\|\boldsymbol{q}_i\|_2} 
		= r(Q).
	\end{equation*}
	Since $r(Q)$ is minimal, we see that ${P}$ is also optimal, i.e., $P\in {\mathcal E}_m$. This completes the proof.

	
	(ii) Let $P\in {\mathcal E}_m$ be an optimal polygon with edge vectors $\{ {\boldsymbol{e}}_1,\dotsc, {\boldsymbol{e}}_m\}$. 
	Without loss of generality, we may assume that  $ {\boldsymbol{e}}_1 = (1,0)^T$ after a proper scaling and rotation. 
	For each $i \in [m]$, define $\alpha_i \in [0, 2\pi)$ as the counterclockwise angle from $ {\boldsymbol{e}}_1$ to $ {\boldsymbol{e}}_i$, and let $\psi_i \in (0, \pi]$ satisfy $\psi_i \equiv \alpha_i \pmod{\pi}$. 
	By Lemma \ref{angles}, we see that $P$ is equilateral and 
	\begin{equation}\label{use_PSI2}
		\{\psi_1, \psi_2, \dotsc,\psi_m\} =
		\left\{\frac{j\pi}{m}\; : \; j\in[m] \right\}.
	\end{equation}
	Hence, there exists a permutation $\{s_1,\ldots,s_{m}\}$ of $[m]$ such that
	$\psi_{s_j} = \frac{j\pi}{m}$ for each $j\in[m]$. 
	Let $I := \{j \in [m]:\alpha_{s_j} = \psi_{s_j} \}$, and let $\boldsymbol{\varepsilon} =(\varepsilon_1,\dotsc,\varepsilon_{m})^T$, where $\varepsilon_j=1$ if $j\in I$ and $\varepsilon_j=-1$ if $j\notin I$. 
	Recall that $ {\boldsymbol{e}}_1 = (1,0)^T$ and $\alpha_i$ is the counterclockwise angle from $ {\boldsymbol{e}}_1$ to $ {\boldsymbol{e}}_i$. Then we have 
	\begin{equation*}
		\begin{aligned}
			{\boldsymbol{e}}_{s_j} 
			= (\cos\alpha_{s_j}, \sin\alpha_{s_j})^T 
			&= (\cos\psi_{s_j}, \sin\psi_{s_j})^T 
			= (\cos\frac{j\pi}{m}, \sin\frac{j\pi}{m})^T=\boldsymbol{\mu}_j \quad \text{for}\;j\in I,\\
			{\boldsymbol{e}}_{s_j} 
			= (\cos\alpha_{s_j}, \sin\alpha_{s_j})^T  
			&= (\cos(\pi+\psi_{s_j}), \sin(\pi+\psi_{s_j}))^T \\
			&= -(\cos\frac{j\pi}{m}, \sin\frac{j\pi}{m})^T=-\boldsymbol{\mu}_j \quad \text{for}\;j\in [m] \setminus  I.
		\end{aligned}
	\end{equation*}
	Therefore, the edge vectors of $P$ are $\{\varepsilon_1\boldsymbol{\mu}_1,\ldots,\varepsilon_m\boldsymbol{\mu}_m\}$.
	
	It suffices to prove that $\boldsymbol{\varepsilon}$ satisfies \eqref{eq:partition2}. By Corollary \ref{Minkowski2}, we have
	\begin{equation*}
		\boldsymbol{0}=\sum_{i=1}^{m}\varepsilon_i\boldsymbol{\mu}_{i}
		=   (\sum_{j=1}^m\varepsilon_j\cos\frac{j\pi}{m}, \sum_{j=1}^m\varepsilon_j\sin\frac{j\pi}{m})^T,
	\end{equation*}
	which implies that 
	\begin{equation*}
		\sum_{j=1}^{m} \varepsilon_j \zeta_m^j 
		=\sum_{j=1}^m\varepsilon_j\cos\frac{j\pi}{m} 
		+\ii \sum_{j=1}^m\varepsilon_j\sin\frac{j\pi}{m}=0.
	\end{equation*}
	This completes the proof.
\end{proof}

\section{Discussion and Future Work}
In this paper, we investigated the optimal frame for phase retrieval. Specifically, in the context of ${\mathbb R}^2$, we established a connection between the optimal frame and the perimeter-maximizing polygon problem, which is a classical research problem in discrete geometry. Leveraging this connection, we were able to characterize all optimal frames with $m$ vectors, provided $m$ possesses an odd factor.
Future research will naturally explore the extension of these results to higher-dimensional and complex domains. 
We anticipate that such investigations will reveal deeper connections between phase retrieval and the fundamental properties of various polytopes, potentially enriching our understanding of both fields.

We conclude this paper by highlighting an intriguing connection between the perimeter-maximizing polygon problem and a fundamental question in discrepancy theory. 
Assume that $E=\{\boldsymbol{e}_1,\dotsc,\boldsymbol{e}_{m}\}$ is the edge set of a convex $m$-gon $ {P} \subset \mathbb{R}^2$ with fixed perimeter $\mathrm{perim}(P)=m$.
Lemma \ref{tight_gon-form} shows that
\begin{equation}
	r( {P}) =\frac{1}{2m} \cdot \underset{\boldsymbol{u} \in \mathbb{S}^1}{\max}\,\sum_{j=1}^{m} {|\langle \boldsymbol{e}_{j},\boldsymbol{u} \rangle|}.
\end{equation}
Hence, solving the perimeter-maximizing polygon problem is equivalent to solving the following optimization problem:
\begin{equation}\label{eq:xu222025}
	\begin{aligned}
		\min_{\substack{\{\boldsymbol{e}_1,\dotsc,\boldsymbol{e}_{m}\}\subset  \mathbb{R}^{2}}}
		\underset{\boldsymbol{u} \in \mathbb{S}^1}{\max}&\,\sum_{j=1}^{m} {|\langle \boldsymbol{e}_{j},\boldsymbol{u} \rangle|}\\
		{\rm s.t.}&\,\,
		\sum_{j=1}^m\boldsymbol{e}_j=\boldsymbol{0}, \,\, \sum_{j=1}^m\|\boldsymbol{e}_j\|_2=m.
	\end{aligned}
\end{equation}
As shown in \cite[Proposition 3]{AS20}, for any $\{\boldsymbol{e}_1,\dotsc,\boldsymbol{e}_{m}\}\subset  \mathbb{S}^{1}$, we have 
\begin{equation*}
	\underset{\boldsymbol{u} \in \mathbb{S}^1}{\max}\,\sum_{j=1}^{m} {|\langle \boldsymbol{e}_{j},\boldsymbol{u} \rangle|}
	=
	\max_{\boldsymbol{\varepsilon}\in \{\pm 1\}^m}\bigg\| \sum_{j=1}^{m}\varepsilon_j\cdot \boldsymbol{e}_{j}\bigg\|_2.
\end{equation*}
Hence, if we assume that each $\|\boldsymbol{e}_j\|=1$, then \eqref{eq:xu222025} can be rewritten as
\begin{equation}
	\label{eq:xu23}
	\begin{aligned}
		\min_{\substack{\{\boldsymbol{e}_1,\dotsc,\boldsymbol{e}_{m}\}\subset  \mathbb{S}^{1}}}
		\max_{\boldsymbol{\varepsilon}\in \{\pm 1\}^m}&\bigg\| \sum_{j=1}^{m}\varepsilon_j\cdot \boldsymbol{e}_{j}\bigg\|_2\\
		{\rm s.t.}&\,\,
		\sum_{j=1}^m\boldsymbol{e}_j=\boldsymbol{0} .
	\end{aligned}
\end{equation}
This means that the problem of minimizing $r(P)$ among all equilateral convex $m$-gons can be reformulated as the  discrepancy problem in \eqref{eq:xu23}.
Understanding how to effectively solve or approximate the perimeter-maximizing polygon problem using discrepancy-theoretic methods presents an exciting direction for future investigation.

\end{document}